\documentclass[11pt]{article}
\usepackage[margin=1in]{geometry}
\usepackage[T1]{fontenc}
\usepackage{lmodern,amsmath,amssymb,amsthm,mathtools}
\usepackage{microtype,needspace}
\usepackage[colorlinks=true,linkcolor=blue,citecolor=blue,urlcolor=blue]{hyperref}
\newtheorem{theorem}{Theorem}[section]
\newtheorem{lemma}[theorem]{Lemma}
\newtheorem{corollary}[theorem]{Corollary}
\newcommand{\F}{\mathbb F}
\newcommand{\PP}{\mathbb P}

\DeclareMathOperator{\RS}{RS}
\DeclareMathOperator{\rank}{rank}
\DeclareMathOperator{\PG}{PG}
\title{Deep holes of full-length Reed--Solomon codes}
\author{Bocong Chen$^1$ and Jing Huang$^2$\thanks{E-mail addresses:
\href{mailto:mabcchen@scut.edu.cn}{\texttt{mabcchen@scut.edu.cn}} (B.~Chen) and
\href{jhuangmath@gzhu.edu.cn}{\texttt{jhuangmath@gzhu.edu.cn}} (J. ~Huang).}\\[0.8em]
\begin{minipage}{\dimexpr\textwidth-2\tabcolsep\relax}
\centering\footnotesize
$^1$School of Mathematics, South China University of Technology, Guangzhou 510641, China\\
$^2$School of Mathematics and Information Science, Guangzhou University, Guangzhou 510006, China
\end{minipage}}
\date{}

\begin{document}
\maketitle
\begin{abstract}
For every prime power $q>19683$, we settle the Cheng--Murray deep-hole
problem for
full-length Reed--Solomon codes in every dimension $2\le k\le q-2$,
with the characteristic-two cases stated explicitly.
For odd $q$, the deep holes are precisely the words of degree $k$.
For even $q$ and $k\ge3$, the same conclusion holds except for
an additional family when $k=q-3$. In even characteristic and dimension two, every deep
hole has a unique expression $a h(X)+bX+c$, where $a\ne0$ and
$h$ is a normalized $o$-polynomial. The proof combines quadratic quotient algebras,
graph-theoretic weighted independent-set inequalities, and planar arc completion.
We also obtain normal rational curve rigidity in the corresponding
nonexceptional dimensions, the projective Reed--Solomon covering
radii, and distinct-coordinate solvability and remainder criteria.
\end{abstract}
\begin{quote}\small
\noindent\textbf{Keywords:} Deep hole; Reed--Solomon code; covering radius;
normal rational curve; arc; $o$-polynomial.\par
\noindent\textbf{2020 Mathematics Subject Classification:}
Primary 94B05; Secondary 51E20, 51E21.
\end{quote}

\section{Introduction}
Let $\F_q$ be the finite field with order $q$ and let $C$ be a linear $[n,k]$ code
over $\F_q$, i.e.,    a $k$-dimensional subspace of $\F_q^n$. The Hamming
distance $d(u,v)$ counts the coordinates at which $u$ and $v$ differ.
For a code $C$, let
\[
 d(u,C)=\min_{c\in C}d(u,c),\qquad
 \rho(C)=\max_{u\in\F_q^n}d(u,C).
\]
The number $\rho(C)$ is the \emph{covering radius}, and a word $u$
with $d(u,C)=\rho(C)$ is a \emph{deep hole}. An $[n,k]$ code is
\emph{maximum distance separable} (MDS) if its minimum distance is
$n-k+1$, attaining the Singleton bound.
We study the full-length, or standard, Reed--Solomon code
\[
 C=\RS(\F_q,k)=\{(f(a))_{a\in\F_q}:f\in\F_q[X],\ \deg f<k\},
 \qquad 2\le k\le q-2.
\]
It is an $[q,k]$ MDS code. Every word  has a unique interpolation
polynomial of degree at most $q-1$, and we use $u$ for both the word
and this polynomial. Interpolation at $k$ positions gives
$d(u,C)\le q-k$. If $\deg u=k$, the root bound gives equality,
giving $\rho(C)=q-k$. We call these degree-$k$ words \emph{canonical},
and write them as
\begin{equation}\label{eq:canonical}
 u(X)=\lambda X^k+f(X),\qquad \lambda\in\F_q^*,\quad\deg f<k.
\end{equation}
We also consider the projective Reed--Solomon code
\begin{equation}\label{eq:prs-definition}
 \operatorname{PRS}(\F_q,k)
 =\{((f(a))_{a\in\F_q},[X^{k-1}]f):\deg f<k\},
\end{equation}
where $[X^j]f$ denotes the coefficient of $X^j$ in $f$.
This is an $[q+1,k]$ MDS code; its last coordinate is the evaluation
at infinity.

The geometric terminology makes the consequences of the classification
transparent. An \emph{$n$-arc} in $\PG(k-1,q)$ is a set of $n\ge k$
points such that representatives of any $k$ of them are linearly
independent. It is \emph{complete} if no further point can be added
while preserving this property. For $m\ge2$, a \emph{normal rational curve} (NRC) in
$\PG(m-1,q)$ is a point set projectively equivalent to
\begin{equation}\label{eq:nrc-definition}
 \{(1:a:\cdots:a^{m-1}):a\in\F_q\}
 \ \cup\ \{(0:\cdots:0:1)\}.
\end{equation}
The columns of a generator matrix of an $[n,k]$ MDS code give an
$n$-arc; the code in \eqref{eq:prs-definition} gives an NRC.
A planar $(q+1)$-arc is called an \emph{oval}; a planar $(q+2)$-arc,
which can occur only for even $q$, is a \emph{hyperoval}.
For even $q\ge4$, we write $\mathcal O_q$ for the normalized
\emph{$o$-polynomials}. Explicitly, $h\in\mathcal O_q$ means that
$h\in\F_q[X]$ has degree at most $q-1$, induces a permutation of
$\F_q$, satisfies $h(0)=0$ and $h(1)=1$, and, for every $s\in\F_q$,
the map
\begin{equation}\label{eq:o-polynomial}
 h_s(z)=\begin{cases}
 (h(s+z)+h(s))/z,&z\ne0,\\
 0,&z=0
 \end{cases}
\end{equation}
is a permutation of $\F_q$. This is the usual polynomial description
of hyperovals: the corresponding point set is
\[
 \{(1:x:h(x)):x\in\F_q\}\cup\{(0:1:0),(0:0:1)\}.
\]
See \cite[Definition~1.1]{CSHyperovals} and
\cite[Theorem~2.1]{KKHyperovals}. In particular, $X^2\in\mathcal O_q$.

Deep holes describe the extremal distance behaviour of a code: they
are precisely the received words for which the nearest codeword is
as far away as possible. Their classification is therefore a natural
structural question in decoding theory. The difficulty of general
Reed--Solomon decoding, established by Guruswami and Vardy
\cite{GV}, makes the additional structure supplied by the full
evaluation set especially relevant. In this setting the root bound
already explains one family of extremal words, and the central
question is whether any others exist. The MDS correspondence gives
a second motivation: the same question asks which extra points can
be attached to $q$ points of an NRC.

Cheng and Murray~\cite{CM} approached deep-hole exclusion by seeking
rational points with distinct coordinates on suitable hypersurfaces.
The full-field conjecture attributed to them asserts, in odd
characteristic, that the canonical words are the only deep holes;
see \cite[Conjecture~1]{ZCL} and
\cite[Conjecture~1]{Kaipa}. Here the evaluation set is
$\F_q$; the related problem with evaluation set $\F_q^*$ has
different exceptional families. Early progress largely concerned
words of prescribed degree. Li and Wan~\cite{LW} used character
sums to obtain error-distance estimates, Liao~\cite{Liao} refined
such estimates, and Cafure, Matera and Privitelli~\cite{CMP}
studied the singularities of the associated symmetric hypersurfaces.
Zhu and Wan~\cite{ZW} combined character sums with a sieve for
distinct coordinates, while Li and Zhu~\cite{LZ} determined error
distances for words of degrees $k+1$ and $k+2$ under their stated
hypotheses. These results excluded substantial families of potential
deep holes, with conditions relating $q$, $k$ and $\deg u$.

A complementary line of work classifies all deep holes in ranges
of dimensions. Zhuang, Cheng and Li
\cite[Theorem~9]{ZCL} used deep-hole trees and additive
combinatorics to classify them over an odd prime field when
$k\ge(p-1)/2$, for arbitrary evaluation sets. Their finite-geometric
argument also proves the full-field classification in odd
characteristic $p$ when $k+1\le p$, and in the range
$3\le q-p+1\le k+1\le q-2$
\cite[Theorem~10]{ZCL}; in particular, the prime-field
case is settled. Kaipa's syndrome correspondence
\cite[Proposition~1]{Kaipa} identifies deep-hole
cosets modulo the code with one-coordinate MDS extensions of the dual code.
Applying the extension theorem of Seroussi and Roth~\cite{SR},
he obtains the classification for arbitrary prime powers when
$k\ge\lfloor(q-1)/2\rfloor$, including the nucleus family in
even characteristic \cite[Theorems~1--2 and Corollary~1]{Kaipa}.

The geometry underlying this approach has its own development.
Segre proved that every oval over a field of odd order is a conic
\cite{Segre55}; this is equivalent to the full-field deep-hole
statement for $k=2$ and odd $q$
\cite[Proposition~2]{Kaipa}. His higher-dimensional work
also shows that, for odd $q\ge5$, every $(q+1)$-arc in $\PG(3,q)$
is a twisted cubic \cite{SegreCurves,BLsurvey}.
For odd squares $q=p^{2h}$, Ball and Lavrauw prove that every
$(q+1)$-arc in $\PG(m-1,q)$ is an NRC when
$m\le\sqrt q-\sqrt q/p+1$ \cite[Corollary~3]{BL}.
Their planar-arc bounds, together with Voloch's bounds for nonsquare
orders \cite[Theorems~2--3]{Voloch91}, supply quantitative control
after projection.

A related extension question asks whether the entire NRC is complete.
This was studied by Storme~\cite{Storme}, and further restrictions on arcs and NRC
subsets were obtained by Ball~\cite{Ball12} and Ball and
De Beule~\cite{BDB}. D\"ur~\cite{Dur94} related this completeness
problem to the covering radius of projective Reed--Solomon codes;
Zhang, Wan and Kaipa~\cite{ZWK} developed the projective deep-hole
problem further. These results explain why planar completion and
MDS extension are effective at opposite ends of the dimension range.
A uniform treatment of the intervening dimensions is needed to
join them.

We supply that treatment for every prime power $q>19683$.
In this range, the following theorem completely resolves the Cheng--Murray
conjecture in odd characteristic for every $2\le k\le q-2$,
and gives its precise all-characteristic counterpart. Two phenomena
must be distinguished in characteristic two: the nucleus family
at $k=q-3$, and the hyperoval description at $k=2$.

\Needspace{18\baselineskip}
\begin{theorem}\label{thm:main}
Let $q$ be a prime power.
\begin{enumerate}
\item If $q>19683$ is odd and $2\le k\le q-2$, all deep holes of $\RS(\F_q,k)$ 
are precisely  of the form \eqref{eq:canonical}.  
\item If $q>19683$ is even and $3\le k\le q-2$, the deep holes are precisely
\eqref{eq:canonical}, together with the additional family
\begin{equation}\label{eq:nucleus-family}
 u(X)=\lambda X^{q-2}+f(X),\qquad\lambda\ne0,\quad\deg f<q-3,
\end{equation}
when $k=q-3$, and only then.
\item If $q\geq4$ is even and $k=2$, a word is a deep hole of
$\RS(\F_q,2)$ if and only if it has the unique form
\[
 u(X)=a h(X)+bX+c,\qquad
 a\in\F_q^*,\quad b,c\in\F_q,\quad h\in\mathcal O_q.
\]
\end{enumerate}
\end{theorem}

Thus the conjectured canonical classification is settled throughout
the stated field-order range, with the characteristic-two cases
specified explicitly. For even $q$ and $k=2$, part~(iii) gives
an exact and unique parametrization by $\mathcal O_q$, valid for
every even $q\ge4$. It does not provide an explicit list of the
$o$-polynomials in $\mathcal O_q$. Determining their possible forms
in general is tied to the classification of hyperovals, a separate
and difficult problem in finite geometry; that classification is
not supplied here. In \eqref{eq:nucleus-family}, the coefficient
of $X^{q-3}$ is zero;
the exceptional family does not consist of all polynomials of
degree $q-2$.

The restriction $q>19683$ is a sufficient bound for joining the
geometric and analytic ranges in parts~(i)--(ii). The proofs in this
paper do not establish the all-dimensional classification for the
remaining field orders. Removing this restriction would require
additional arguments or complete, independently verifiable
computational certificates. No such computational extension is
asserted here. Part~(iii), however, holds for every even $q\ge4$,
as proved directly in Section~\ref{sec:dimension-two}.

The proof combines the two approaches described above. A
noncanonical deep hole yields many distinct nonsingular quartic
functionals. Their quadratic contractions define projective
involutions, whose orbits form graphs on the field. The weighted
independent-set inequality of Sah, Sawhney, Stoner and Zhao
\cite[Theorem~1.5]{SSSZ} bounds the number of these
functionals that survive a small set. Multiplicity estimates and
averaging then exclude the intermediate dimensions. In the lower
dimensions, prescribed planar projections and a common annihilator
identity transfer conic or hyperoval completion back to the word.
Explicit inequalities make the two ranges overlap for $q>19683$.

We record four consequences to make the geometric and polynomial
content of the theorem explicit. Their proofs are collected in
Appendix~\ref{app:corollaries}. The first uses the syndrome
correspondence of \cite[Proposition~1]{Kaipa} and gives the NRC
rigidity statement in its applicable dimensions.
\begin{corollary}\label{cor:nrc-rigidity}
Let $q>19683$ and $2\le k\le q-2$, and let $\mathcal A$ be a
$(q+1)$-arc in $\PG(k-1,q)$ containing $q$ points of an NRC
$\mathcal N$. If $q$ is odd, or if $q$ is even and
$k\notin\{3,q-2\}$, then $\mathcal A=\mathcal N$.
If $q$ is even and $k=3$, the extra point is either the missing
point of the conic or its nucleus, the common point of its tangents.
\end{corollary}

The exclusion $k=q-2$ in even characteristic corresponds to
part~(iii) of the theorem and cannot be suppressed. For example,
when $q=2^{15}>19683$, the polynomial $X^4$ belongs to
$\mathcal O_q$, since the maps in \eqref{eq:o-polynomial} are
$z\mapsto z^3$ and $\gcd(3,2^{15}-1)=1$. Its noncanonical deep hole gives a completion
outside the NRC in that dimension. 

Completeness of the whole NRC
is a weaker question and still has the following uniform answer.
D\"ur's criterion, in the form stated in
\cite[Section~IV]{Kaipa}, translates it into a covering
radius statement.
\begin{corollary}\label{cor:prs-radius}
Let $q>19683$ and $2\le k\le q-2$. Then
\[
 \rho\bigl(\operatorname{PRS}(\F_q,k)\bigr)
 =\begin{cases}
 q+1-k,&q\text{ even and }k\in\{2,q-2\},\\
 q-k,&\text{otherwise}.
 \end{cases}
\]
Equivalently, for $2\le k\le q-1$, every NRC in
$\PG(k-1,q)$ is complete, except when $q$ is even and
$k\in\{3,q-1\}$.
\end{corollary}

The next consequence expresses deep-hole exclusion as a solvability
statement. Write $e_0=1$ and, for $1\le r\le k$, let
\[
 e_r(x_1,\ldots,x_k)
 =\sum_{1\le i_1<\cdots<i_r\le k}x_{i_1}\cdots x_{i_r}
\]
be the elementary symmetric polynomials. The restriction below
removes exactly the two dimensions in which the main theorem
allows noncanonical deep holes.
\begin{corollary}\label{cor:symmetric-solutions}
Let $q>19683$ and $2\le k\le q-3$; if $q$ is even, assume
$k\notin\{2,q-3\}$. Let $0\le\ell<k$ and
$(\delta_0,\ldots,\delta_k)\in\F_q^{k+1}$ satisfy
$\delta_0=\cdots=\delta_{\ell-1}=0$ and $\delta_\ell\ne0$.
Then
\[
 \sum_{r=0}^{k-\ell}(-1)^r\delta_{k-r}
 e_r(x_1,\ldots,x_k)=0
\]
has a solution in $\F_q^k$ with pairwise distinct coordinates.
In particular, $e_r(x_1,\ldots,x_k)=0$ has such a solution for
every $1\le r\le k$.
\end{corollary}

Finally, agreement with a polynomial of degree less than $k$
can be expressed through a remainder. This gives a second
polynomial formulation, now with all exceptional words specified.
\begin{corollary}\label{cor:small-remainder}
Let $q>19683$, $2\le k\le q-2$, and
$k+1\le\deg u\le q-1$. There is a $(k+1)$-set
$B\subseteq\F_q$ such that the remainder of $u(X)$ on division
by $\prod_{b\in B}(X-b)$ has degree at most $k-1$, except
precisely in the following cases:
\begin{enumerate}
\item $q$ is even, $k=q-3$, and $u$ has the form
\eqref{eq:nucleus-family};
\item $q$ is even, $k=2$, and $u=a h+bX+c$ for some
$a\in\F_q^*$, $b,c\in\F_q$ and $h\in\mathcal O_q$.
\end{enumerate}
In particular, such a set always exists when $q$ is odd, or when
$q$ is even and $k\notin\{2,q-3\}$.
\end{corollary}

These consequences settle the NRC extension and projective
covering-radius questions in the stated range, and turn the
classification into uniform distinct-coordinate solvability
criteria.

Section~\ref{sec:preliminaries} collects the notation and the
published inputs used below. Section~\ref{sec:contractions}
proves the dimension-two statement and the multiplicity bound
for quartic contractions. Section~\ref{sec:graphs} constructs
the involution detectors. Section~\ref{sec:geometry} joins the
planar bounds with the analytic estimates and completes the
proof of Theorem~\ref{thm:main}. Appendix~\ref{app:corollaries}
contains the proofs of the four corollaries stated above.

\section{Preliminaries}\label{sec:preliminaries}
We introduce the polynomial and geometric notation used throughout
the paper, and recall the deep-hole, planar-arc and graph estimates
needed below. All logarithms are natural.
Sets are denoted by uppercase letters and their elements by lowercase
letters, as in $b\in B$. Linear functionals are denoted by uppercase
calligraphic letters. Put $\mathcal P_d=\F_q[X]_{\le d}$; for $d<q$ we
also regard this space as its evaluation space on $\F_q$.
For $A\subseteq\F_q$, write $p_A(X)=\prod_{a\in A}(X-a)$.
Throughout the contraction argument, set
\begin{equation}\label{eq:functional}
 t=q-k-1,\qquad
 \mathcal L(X^j)=-[X^{q-1-j}]u\quad(0\le j\le t).
\end{equation}
For a functional $\mathcal M$ defined through degree $a+b$, write
$H_{a,b}(\mathcal M)=(\mathcal M(X^{i+j}))_{0\le i\le a,\,0\le j\le b}$;
for a quartic functional, abbreviate $H=H_{2,2}$. A nonzero multiple of
$[X^t]$ is canonical. We use
\begin{equation}\label{eq:characteristic-constants}
 \iota_q=\begin{cases}0,&q\text{ odd},\\1,&q\text{ even},\end{cases}
 \qquad\delta_q=2-\iota_q,\qquad
 H_q=\left\lfloor\frac{q-5}{2}\right\rfloor.
\end{equation}
Let
\begin{equation}\label{eq:fiber-budget-notation}
 Q_{\rm f}(q,k)=H_q\bigl(\lfloor q/2\rfloor-k\bigr).
\end{equation}
Homogenization identifies $\mathcal P_d$ with the binary forms
$\mathcal H_d$ of degree $d$ in $X,Y$; a factor at infinity records a
drop in ordinary degree. For a planar point set $A$, write $I(A)_j$ for the space of
homogeneous degree-$j$ forms vanishing on $A$.
B\'ezout arguments are made over an algebraic closure when needed.

The next characterization is due to Kaipa
\cite[Proposition~1 and equations~(5)--(7)]{Kaipa}. We include its
interpolation proof to fix the functional normalization.
\begin{lemma}
\label{lem:mds}
A word $u$ is deep if and only if
\begin{equation*}\tag{NV}\label{eq:NV}
 \mathcal L(p_A)\ne0\qquad\text{for every }A\in\binom{\F_q}{t}.
\end{equation*}
Equivalently, the columns with rows
$((a^j)_{a\in\F_q},0)$, $0\le j<k$, and
$((u(a))_{a\in\F_q},1)$ form a $(q+1)$-arc in $\PG(k,q)$.
\end{lemma}
\begin{proof}
A word is deep exactly when it agrees with no polynomial of degree less
than $k$ on a $(k+1)$-set $B$. Lagrange interpolation expresses this as
$\sum_{b\in B}u(b)/p_B'(b)\ne0$. For $A=\F_q\setminus B$,
differentiating $p_Ap_B=X^q-X$ gives $p_A(b)=-1/p_B'(b)$.
Since $\deg(up_A)<2(q-1)$, the field power sums give
$$
 \mathcal L(p_A)=\sum_{x\in\F_q}u(x)p_A(x)
 =-\sum_{b\in B}\frac{u(b)}{p_B'(b)}.
$$
This proves the criterion. Minors using the last column are Vandermonde
minors; the other minors give the same interpolation condition. 
\end{proof}

The high-dimensional classification is given by Kaipa
\cite[Theorems~1--2 and Corollary~1]{Kaipa}. The characteristic-dependent
range is due to Zhuang, Cheng and Li
\cite[Theorem~10]{ZCL}.
\begin{lemma}\label{lem:high}\label{lem:prime}
The classifications in parts~(i)--(ii) of Theorem~\ref{thm:main}
hold for $k\ge\lfloor(q-1)/2\rfloor$. In odd characteristic $p$,
the canonical classification also holds for $k+1\le p$, and hence
over every odd prime field.
\end{lemma}

For the lower dimensions we use planar arc bounds. The following
curve-containment theorem is due to Ball and Lavrauw
\cite[Theorem~1]{BL}.
\begin{lemma}\label{lem:two-curves}
Let $q=p^m$ be odd and let $A$ be a nonconic planar arc of size
$q+2-\tau$. Set $h=p^{\lfloor\log_p\tau\rfloor}$.
If $I(A)_\tau=0$, two coprime curves of degrees at most $\tau+h$
contain $A$. If a specified degree-$\tau$ curve contains $A$ and
\begin{equation}\label{eq:technical}
 h(\tau+h/2+3/2)\le(\tau+2)(\tau+1)/2,
\end{equation}
then a curve of degree at most $\tau+h$, coprime to that curve, contains
$A$.
\end{lemma}

For square fields, Ball and Lavrauw \cite[Theorem~2]{BL} obtain
the following stronger conic-containment bound.
\begin{lemma}\label{lem:square}
For odd $q=p^{2m}$, every planar arc of size at least
$q-\sqrt q+\sqrt q/p+3$ is conic.
\end{lemma}

For odd nonsquare fields, Voloch's complete-arc bound
\cite[Theorem~3]{Voloch91}, followed by maximal extension,
gives the next statement.
\begin{lemma}\label{lem:voloch}
For odd $q=p^{2m+1}$, $m\ge1$, every planar arc of size greater than
$q-\sqrt{pq}/4+29p/16+1$ is contained in a conic.
\end{lemma}

In even characteristic, Segre's complete-arc bound as recalled in
\cite[equation~(1)]{Voloch91}, and Voloch's improvement for
nonsquare fields \cite[Theorem~2]{Voloch91}, yield the
following completion statement. The hyperoval need not be regular.
\begin{lemma}\label{lem:hyperoval-completion}
For even $q\ge4$, every planar arc of size greater than
$q-\sqrt q+1$ extends to a hyperoval. For nonsquare $q\ge8$, the
threshold can be replaced by $q-\sqrt{2q}+2$.
\end{lemma}

Recall that an \emph{independent set} in a graph is a set of vertices, no two of which are adjacent. 
The detector argument uses the following weighted independent-set
inequality of Sah, Sawhney, Stoner and Zhao
\cite[Theorem~1.5]{SSSZ}.
\begin{lemma}\label{lem:sssz}
For a graph $G$ without isolated vertices, with degrees $d_v$, and
$\lambda>0$,
$$
 P_G(\lambda)=\sum_{I\text{ independent}}\lambda^{|I|}
 \le\prod_{uv\in E(G)}
 \big((1+\lambda)^{d_u}+(1+\lambda)^{d_v}-1\big)^{1/(d_ud_v)}.
$$
An isolated vertex contributes $1+\lambda$.
\end{lemma}

\section{Dimension two and contraction multiplicities}
\label{sec:contractions}\label{sec:dimension-two}
We first prove the unique normalization in dimension two over even
fields. We then treat the higher dimensions by contracting the
nonvanishing functional from \eqref{eq:NV} to quadratic and quartic
polynomial spaces.
For the dimension-two argument let $q=2^m\ge4$. The description by
$o$-polynomials is the classical hyperoval description; see
\cite[Definition~1.1]{CSHyperovals} and
\cite[Theorem~2.1]{KKHyperovals}. The following slope calculation
gives a direct proof of the normalization.
\begin{lemma}\label{lem:common-missing-slope}
Suppose the graph of a function $u:\F_q\to\F_q$ contains no three
collinear points. There is a unique $b\in\F_q$ such that, for every
$x\in\F_q$,
\[
 \left\{\frac{u(y)+u(x)}{y+x}:y\in\F_q\setminus\{x\}\right\}
 =\F_q\setminus\{b\}.
\]
Moreover, the interpolation polynomial of $u$ satisfies $u'(X)=b$.
\end{lemma}
\begin{proof}
For distinct $x,y$, put $d_{xy}=(u(x)+u(y))/(x+y)$.
For fixed $x$, the $q-1$ values $d_{xy}$ are distinct, since an
equality would give three collinear graph points. Write $b_x$ for
the unique missing value. Fix distinct $x,y$ and put $d=d_{xy}$.
For $z\notin\{x,y\}$, direct subtraction gives
\[
 \frac{d+d_{xz}}{d+d_{yz}}=\frac{y+z}{x+z}.
\]
Both sides are defined, since there are no three collinear points.
Differentiating $\prod_{a\in\F_q}(X+a)=X^q+X$ at $x$ and $y$ gives
\[
 \prod_{z\notin\{x,y\}}(x+z)
 =\prod_{z\notin\{x,y\}}(y+z)=\frac1{x+y}.
\]
Consequently the products of $d+d_{xz}$ and $d+d_{yz}$ over
$z\notin\{x,y\}$ are equal. On the other hand,
\[
 \prod_{z\notin\{x,y\}}(d+d_{xz})
 =\prod_{a\in\F_q\setminus\{b_x,d\}}(d+a)
 =\frac1{d+b_x},
\]
because $\prod_{a\ne d}(d+a)=1$. The corresponding formula at $y$
therefore gives $b_x=b_y$. Denote this common value by $b$.

For the derivative assertion, expand the interpolation polynomial as
$u(x+Z)+u(x)=\sum_{j=1}^{q-1}c_j(x)Z^j$, where $c_1(x)=u'(x)$.
The finite-field power sums give
\[
 \sum_{z\in\F_q^*}\frac{u(x+z)+u(x)}z=u'(x),
\]
since $\sum_{z\ne0}1=1$ in characteristic two and
$\sum_{z\ne0}z^j=0$ for $1\le j\le q-2$.
The slopes on the left run through $\F_q\setminus\{b\}$.
As $\sum_{a\in\F_q}a=0$ for $q\ge4$, their sum is $b$.
Thus $u'(x)=b$ at every $x$. Since $\deg(u'-b)\le q-2$, this is
the polynomial identity $u'=b$.
\end{proof}

Subtracting the common missing slope produces a permutation. Its
values at zero and one then fix the remaining affine normalization.
\begin{lemma}
\label{thm:dimension-two}
A word $u$ is a deep hole of $\RS(\F_q,2)$ if and only if
\begin{equation}\label{eq:oval-normal-form}
 u(X)=a h(X)+bX+c,\qquad
 a\in\F_q^*,\quad b,c\in\F_q,\quad h\in\mathcal O_q.
\end{equation}
This representation is unique. Its parameters are
\begin{equation}\label{eq:oval-normalization}
 b=[X]u,\qquad c=u(0),\qquad a=u(1)+b+c,
 \qquad h(X)=a^{-1}\bigl(u(X)+bX+c\bigr).
\end{equation}
In particular, every such word has no odd-degree term other than
possibly $bX$, and the number of deep words is
$q^2(q-1)|\mathcal O_q|$.
\end{lemma}
\begin{proof}
The covering radius of $\RS(\F_q,2)$ is $q-2$.
Thus $u$ is deep exactly when no affine linear polynomial agrees
with $u$ at three distinct arguments, or equivalently when its graph
contains no three collinear points.

Suppose first that $u$ is deep, and let $b$ be the missing slope in
Lemma~\ref{lem:common-missing-slope}. Put $c=u(0)$ and
$g(X)=u(X)+bX+c$. For $x\ne y$, an equality $g(x)=g(y)$ would
give the forbidden slope $d_{xy}=b$. Hence $g$ is a permutation,
$g(0)=0$, and $a=g(1)\ne0$. Set $h=g/a$.
For each fixed $s$, the values
\[
 \frac{h(s+z)+h(s)}z
 =a^{-1}\left(\frac{u(s+z)+u(s)}z+b\right),\qquad z\ne0,
\]
run through $\F_q^*$. Together with value zero at $z=0$, they
form a permutation. Therefore $h\in\mathcal O_q$, proving
\eqref{eq:oval-normal-form}.

Conversely, suppose \eqref{eq:oval-normal-form} holds.
For each $s$, the slopes from $(s,u(s))$ to the other graph points
are $b+a h_s(z)$ for $z\ne0$, and are pairwise distinct by
\eqref{eq:o-polynomial}. Thus the graph has no three collinear
points and $u$ is deep.

For every normalized $o$-polynomial $h$, its graph has common
missing slope zero; Lemma~\ref{lem:common-missing-slope} gives
$h'=0$. Therefore every positive exponent occurring in $h$ is
even, and $\deg h\le q-2$. In any representation
\eqref{eq:oval-normal-form}, differentiation and evaluation at
$0,1$ now force precisely \eqref{eq:oval-normalization}.
This proves uniqueness and the coefficient assertion. Counting the
independent choices of $a,b,c,h$ proves the final formula.
\end{proof}

We now return to arbitrary characteristic and to the functional
$\mathcal L$ defined in \eqref{eq:functional}. The first step is
to find a nonsingular quadratic contraction outside any prescribed
set of at most $k+2$ points.
\begin{lemma}\label{lem:quadratic-seed}
Suppose $\mathcal{L}$ satisfies \eqref{eq:NV}, is noncanonical, and $t \ge 2$.
For every subset $A \subseteq \F_q$ with $|A| \le k+2$, there exists a  subset
$U \subseteq \F_q \setminus A$ of size $t-2$ such that the associated $2 \times 2$ 
Hankel matrix is nonsingular, i.e.,
\[
 \det H_{1,1}(\mathcal{L}(p_U \,\cdot\,)) = 
\det \begin{pmatrix} \mathcal{L}(p_U) & \mathcal{L}(Xp_U) \\ 
\mathcal{L}(Xp_U) & \mathcal{L}(X^2p_U) \end{pmatrix} \ne 0.
\]
\end{lemma}
\begin{proof}
Choose $B\subseteq\F_q\setminus A$ with $|B|=t-1$,
and write $c_0=\mathcal{L}(p_B)$, $c_1=\mathcal{L}(Xp_B)$ and
$h_b=\mathcal{L}(p_B/(X-b))$ for $b\in B$. For $U=B\setminus\{b\}$,
we have
\[
\mathcal{L}(p_U)=h_b,\quad  \mathcal{L}(Xp_U)=bh_b+c_0,\quad  \mathcal{L}(X^2p_U)=b^2h_b+bc_0+c_1,
\]
and their quadratic determinant is
$ \det H_{1,1}(\mathcal{L}(p_U \,\cdot\,))=h_b(c_1-bc_0)-c_0^2$. Suppose all these determinants vanish.
The polynomials $p_B/(X-b)$, for $b\in B$, form a basis of
$\F_q[X]_{\le t-2}$.
If $c_0 = c_1 = 0$, choose an arbitrary element $y \notin B$. Applying the functional $\mathcal{L}$ yields
\[
\mathcal{L}(p_B \cdot (X-y)) = \mathcal{L}(X p_B) - y \mathcal{L}(p_B) = c_1 - y c_0 = 0.
\]
This produces a polynomial of degree $t$ fully splitting over $\F_q$ 
in the kernel of $\mathcal{L}$, contrary to \eqref{eq:NV}.
If $c_0 = 0$ and $c_1 \ne 0$, the vanishing determinant forces $c_1 h_b = 0$, implying all $h_b$ vanish. 
Because $\{p_B/(X-b)\}_{b\in B}$ forms a basis for $\F_q[X]_{\le t-2}$, $\mathcal{L}$ 
must vanish identically on $\F_q[X]_{\le t-2}$. 
Furthermore, $\mathcal{L}(p_B) = c_0 = 0$ forces $\mathcal{L}$ to 
vanish identically on $\F_q[X]_{\le t-1}$. Given $\mathcal{L}(X^t) = c_1 \ne 0$, 
it follows that $\mathcal{L}$ extracts exactly the leading coefficient, i.e., $\mathcal{L} = c_1 [X^t]$. 
This strictly contradicts the noncanonicity of $\mathcal{L}$.

Finally, suppose $c_0 \ne 0$ and define $z = c_1 / c_0$. The zero determinant relation $h_b(c_1 - b c_0) = c_0^2$ 
simplifies to $h_b c_0 (z - b) = c_0^2$. 
This implies $z \notin B$ and uniquely determines $h_b = c_0 / (z - b)$ for every $b \in B$. 
Define a scaled point-evaluation functional $\mathcal{L}^* = c_0 p_B(z)^{-1} \operatorname{ev}_z$,
where $\operatorname{ev}_z$ denotes the evaluation functional mapping $f \mapsto f(z)$.  
We test $\mathcal{L}^*$ on the basis polynomials of $\F_q[X]_{\le t-2}$:
\[
\mathcal{L}^*\left( \frac{p_B(X)}{X-b} \right) =
c_0 p_B(z)^{-1} \frac{p_B(z)}{z-b} = \frac{c_0}{z-b} = h_b = \mathcal{L}\left( \frac{p_B(X)}{X-b} \right).
\]
Furthermore, evaluating $\mathcal{L}^*$ on the higher-degree polynomials 
yields $\mathcal{L}^*(p_B) = c_0 = \mathcal{L}(p_B)$ 
and $\mathcal{L}^*(X p_B) = z c_0 = c_1 = \mathcal{L}(X p_B)$. The collection
\[
 \left\{ \frac{p_B(X)}{X-b} \right\}_{b \in B} \cup \{ p_B, X p_B \}
\]
forms a  basis for the vector space $\F_q[X]_{\le t}$. 
Since $\mathcal{L}^*$ and $\mathcal{L}$ agree on this entire basis, they must be 
identical on $\F_q[X]_{\le t}$. Hence, $\mathcal{L}$ is merely a finite evaluation up 
to a nonzero scalar, again contradicting \eqref{eq:NV}.
\end{proof}

Omitting one point from such a contraction gives a one-dimensional
kernel. Comparing these kernels detects when a two-point omission
produces a nonsingular quartic contraction.

\begin{lemma}\label{lem:kernel-omission}
Let $U\subseteq\F_q$ have size $n$, and let $\mathcal R$  be a linear functional defined on the polynomial space
$\F_q[X]_{\le n+2}$. Suppose $H_{1,1}(\mathcal R(p_U\,\cdot))$
is nonsingular. For $u\in U$, let
$\mathcal M_u=\mathcal R(p_U/(X-u)\,\cdot)$, and let $[Q_u]$
be the one-dimensional kernel of $H_{1,2}(\mathcal M_u)$.
For distinct $u_1,u_2\in U$,
$$
 \det H_{2,2}\!\left(\mathcal R\left(
 \frac{p_U}{(X-u_1)(X-u_2)}\,\cdot\right)\right)=0
 \quad\hbox{if and only if}\quad [Q_{u_1}]=[Q_{u_2}].
$$
Furthermore, define an equivalence relation on $U$ by $u_1 \sim u_2$
precisely when  $[Q_{u_1}] = [Q_{u_2}]$.
If $C\subseteq U$ is an equivalence class of size $m$ and
$\mathcal M=\mathcal R(p_{U\setminus C}\,\cdot)$, then
$\rank H_{2,m}(\mathcal M)\le2$.
\end{lemma}
\begin{proof}
View $H_{1,2}(\mathcal M_u)$ as a linear map from $\mathbb{F}_q[X]_{\le 2}$ to $\mathbb{F}_q^2$,
mapping $f(X)$ to the column vector $(\mathcal M_u(f), \mathcal M_u(Xf))^{\mathsf T}$.
Restriction of this map to the two-dimensional subspace of polynomials divisible by $X-u$,
namely $\{ (X-u)g(X) \mid g \in \mathbb{F}_q[X]_{\le 1} \}$, yields the $i$-th output component
$$
 \mathcal M_u(X^i(X-u)g(X)) 
 = \mathcal R\left( \frac{p_U(X)}{X-u} X^i(X-u)g(X) \right)
 = \mathcal R\left( p_U(X) X^i g(X) \right) \quad (i=0,1).
$$
This algebraically recovers the action of the given nonsingular quadratic matrix 
$H_{1,1}(\mathcal R(p_U\,\cdot))$ on $g(X)$. Consequently, the image of $H_{1,2}(\mathcal M_u)$ 
has dimension at least two. Since its codomain is $\mathbb{F}_q^2$, its rank is exactly two.

Put $\mathcal F=\mathcal R(p_U/((X-u_1)(X-u_2))\,\cdot)$. 
By definition, we have the algebraic identities 
$\mathcal M_{u_1}(f) = \mathcal F(Xf) - u_2\mathcal F(f)$ 
and $\mathcal M_{u_2}(f) = \mathcal F(Xf) - u_1\mathcal F(f)$ for every 
$f\in\F_q[X]_{\le3}$.
If $\det H_{2,2}(\mathcal F)=0$, there exists a nonzero polynomial 
$Q \in \ker H_{2,2}(\mathcal F)$, 
meaning $\mathcal F(X^i Q) = 0$ for $i=0,1,2$. 
The identities then yield 
$\mathcal M_{u_1}(X^i Q) = \mathcal M_{u_2}(X^i Q) = 0$ for $i=0,1$. 
Thus, $Q$ is a non-trivial element of both one-dimensional kernels 
$[Q_{u_1}]$ and $[Q_{u_2}]$, which forces $[Q_{u_1}] = [Q_{u_2}]$.

Conversely, if $[Q_{u_1}] = [Q_{u_2}]$, a common nonzero kernel polynomial $Q$ satisfies 
$\mathcal M_{u_1}(X^i Q) = \mathcal M_{u_2}(X^i Q) = 0$ for $i=0,1$. 
The identity $(u_1-u_2)\mathcal F = \mathcal M_{u_1} - \mathcal M_{u_2}$ and $u_1 \ne u_2$ 
immediately yield $\mathcal F(Q) = \mathcal F(XQ) = 0$. 
Then $\mathcal M_{u_1}(XQ) = \mathcal F(X^2Q) - u_2\mathcal F(XQ) = 0$ 
forces $\mathcal F(X^2Q) = 0$. 
Thus $Q \in \ker H_{2,2}(\mathcal F)$, yielding a zero determinant.

Finally, let $C \subseteq U$ be an equivalence class of size $m$ 
with common kernel $Q(X) = \sum_{j=0}^2 q_j X^j$. 
For every $c \in C$ and $i \in \{0,1\}$, $Q \in [Q_c]$ implies
$$
 0 = \mathcal M_c(X^i Q) = 
\mathcal R\left( p_{U \setminus C} \frac{p_C}{X-c} X^i Q \right) 
= \mathcal M\left( Q X^i \frac{p_C}{X-c} \right).
$$
The polynomials $X^i p_C/(X-c)$ for $c \in C$ and 
$i \in \{0,1\}$ span $\F_q[X]_{\le m}$. 
Linearity therefore gives 
$\mathcal M(X^k Q) = \sum_{j=0}^2 q_j \mathcal M(X^{j+k}) = 0$ for all $0 \le k \le m$. 
This implies the nonzero coefficient vector $(q_0, q_1, q_2)$ 
is a left kernel vector of the $3 \times (m+1)$ matrix $H_{2,m}(\mathcal M)$, 
forcing $\rank H_{2,m}(\mathcal M) \le 2$.
\end{proof}

To control the sizes of the resulting kernel classes, we pass to
the two-dimensional algebra determined by a quadratic kernel form.

\begin{lemma}\label{lem:dense-quad}
Let $q\ge16$, $k\ge2+\iota_q$, and let $C,V\subseteq\F_q$ be
disjoint, with $|C|=m\ge3$ and $|V|=k+3$. Put $r=m+2$.
Suppose $\mathcal M$ acts on $\F_q[X]_{\le r}$ and satisfies
$$
 \rank H_{2,m}(\mathcal M)\le2,\qquad
 \det H_{1,1}(\mathcal M(p_C\,\cdot))
 =\mathcal M(p_C)\mathcal M(X^2p_C)-\mathcal M(Xp_C)^2\ne0,
$$
and $\mathcal M(p_S)\ne0$ for every $S\in\binom{C\cup V}{r}$.
Then $m\le H_q=\lfloor\frac{q-5}{2}\rfloor$.
\end{lemma}
\begin{proof}
We begin by formally homogenizing the polynomial space. Let $\mathcal{H}_d$ 
denote the vector space of binary forms in $X$ and $Y$ of degree $d$.
For $0\le d\le r$ and $F\in\mathcal H_d$, write
$\mathcal M(F)=\mathcal M(F(X,1))$; in particular,
$\mathcal M(X^aY^b)=\mathcal M(X^a)$ when $a+b\le r$.
The matrix $H_{2,m}(\mathcal M)$ has three rows with entries $\mathcal M(X^{i+j})$.
Its rank deficiency implies the existence of a nonzero vector $(q_0, q_1, q_2)$ in the left nullspace.
We use these coefficients to define a nonzero quadratic form $Q(X,Y) = q_0 Y^2 + q_1 XY + q_2 X^2 \in \mathcal H_2$,
giving
\[
 \mathcal M(Q(X,Y) \cdot X^j Y^{m-j}) = \sum_{i=0}^2 q_i 
\mathcal M(X^{i+j} Y^{m-j+2-i}) = \sum_{i=0}^2 q_i \mathcal M(X^{i+j}) = 0.
\]
Since $\{X^jY^{m-j}:0\le j\le m\}$ is a basis of $\mathcal H_m$,
linearity gives $\mathcal M(QR)=0$ for every $R\in\mathcal H_m$;
equivalently, $\mathcal M(Q\mathcal H_m)=0$.

Choose a linear form $\ell$ that does not divide $Q$, 
and a complementary linear form $\eta$ 
such that $\{\eta, \ell\}$ forms a basis for $\mathcal{H}_1$. 
Every form $F \in \mathcal{H}_j$ can be uniquely rewritten as a 
homogeneous polynomial $\widetilde{F}(\eta, \ell)$ of degree $j$. 
Dehomogenization with respect to $\ell$ defines the polynomial
$\chi(Z)=\widetilde Q(Z,1)$.
Since $\ell\nmid Q$, the coefficient of $\eta^2$ in
$\widetilde Q$ is nonzero, so $\deg\chi=2$.
Define the two-dimensional quotient algebra
\[
 \mathbb A=\F_q[Z]/(\chi(Z)),
\]
and write $z$ for the residue class of $Z$ in $\mathbb A$.
For every $j\ge0$, define
\[
 \rho_j:\mathcal H_j\longrightarrow\mathbb A,\qquad
 F\longmapsto\widetilde F(z,1).
\]
Since $r=m+2\ge5$, the images $\rho_r(\ell^r)=1$ and
$\rho_r(\eta\ell^{r-1})=z$ demonstrate that $\rho_r$ is surjective.
Its kernel contains $Q\mathcal H_{r-2}$, a subspace of dimension 
$r-1$. Because the codomain $\mathbb A$ has dimension two, 
the rank-nullity theorem forces the kernel to be exactly $Q\mathcal H_{r-2}$.

The annihilation hypothesis $\mathcal M(Q\mathcal H_m) = 0$ therefore dictates that $\ker\rho_{m+2} \subseteq \ker\mathcal M$. 
By the universal property of quotient spaces, $\mathcal M$ 
factors through $\rho_{m+2}$, meaning there exists a unique linear 
functional $\mathcal E: \mathbb A \to \F_q$ such that $\mathcal M = \mathcal E \circ \rho_{m+2}$.

For each  $a\in \F_q$, denote its projection by $\xi_a = \rho_1(X-aY)$. 
It follows that 
the degree-$m$ polynomial 
$p_C(X,Y) = \prod_{c\in C}(X-cY)$ projects to $\xi_C = \prod_{c\in C}\xi_c$.
By hypothesis, the $2 \times 2$ matrix $H_{1,1}(\mathcal M(p_C\,\cdot))$ is nonsingular. 
Homogenizing its $(i,j)$-th entry to total degree $m+2$ yields $\mathcal M(p_C \cdot X^{i+j} Y^{2-i-j})$. 
Applying the factorization $\mathcal M = \mathcal E \circ \rho_{m+2}$ and expanding the product inside $\mathbb A$ gives:
$$ \mathcal M(p_C X^{i+j} Y^{2-i-j}) = \mathcal E\Big( \rho_m(p_C) \rho_1(X)^{i+j} \rho_1(Y)^{2-i-j} \Big) = 
\mathcal E\Big( \xi_C \rho_1(X)^{i+j} \rho_1(Y)^{2-i-j} \Big). $$
Since the projections $\{\rho_1(Y), \rho_1(X)\}$ span the two-dimensional algebra $\mathbb A$, these entries exactly form 
the Gram matrix of the bilinear pairing $(x,y) \mapsto \mathcal E(\xi_C xy)$ on $\mathbb A \times \mathbb A$.
The nonsingularity ensures this pairing is nondegenerate. This forces $\xi_C$ to be a unit in $\mathbb A$ (otherwise, 
multiplication by a zero divisor $\xi_C$ would create a degenerate pairing), and it simultaneously guarantees that the 
base pairing $(x,y) \mapsto \mathcal E(xy)$ is also nondegenerate.

Because the functional $\mathcal E$ maps a 2-dimensional space onto a 1-dimensional space, its kernel is a 1-dimensional subspace. 
Every nonzero element of $\ker\mathcal E$ must be a unit: if a nonzero non-unit $x$ belonged to $\ker \mathcal E$, 
its principal ideal $\langle x \rangle$ would also be 1-dimensional, forcing 
$\langle x \rangle = \ker\mathcal E$. 
This implies $\mathcal E(xy) = 0$ for all $y \in \mathbb A$, directly contradicting the nondegeneracy of the pairing. 
Fix one such unit $\kappa \in \ker \mathcal E$.

Remove from $C\cup V$ the points for which $\xi_a$ is a nonunit, and denote the remaining set by $W$. 
At most one point is removed: two distinct nonunit factors correspond to the roots of a split $Q$, and any $r$-set 
containing both would yield $p_S \in Q\mathcal H_{r-2}$. 
This forces $\mathcal M(p_S)=0$, directly contradicting the nonvanishing hypothesis. 
Since $\rho_1$ is injective on $\mathcal H_1$, the projective images $[\xi_w]$ for $w\in W$ are pairwise distinct in $G = \mathbb A^*/\F_q^*$. 
For any $S\in\binom Wr$, the factorization $\mathcal M = \mathcal E \circ \rho_r$ yields $\mathcal M(p_S) = \mathcal E(\prod_{a\in S} \xi_a)$. 
Because $\ker \mathcal E$ is the 1-dimensional space spanned by the unit $\kappa$, 
this evaluation vanishes if and only if the product $\prod_{a\in S} \xi_a$ is a scalar multiple of $\kappa$. 
In the quotient group $G$, this is precisely the equivalence:
\begin{equation}\label{eq:unit-product-zero}
 \mathcal M(p_S)=0 \quad\Longleftrightarrow\quad \prod_{a\in S}[\xi_a]=[\kappa].
\end{equation}
Depending on whether $Q$ is split, irreducible, or repeated, the projective group $G = \mathbb{A}^*/\F_q^*$ 
is isomorphic to $\F_q^*$, $\F_{q^2}^*/\F_q^*$, or the additive group $(\F_q,+)$, respectively. In all cases, $|G| \le q+1$. 

To construct a subset $S \in \binom{W}{r}$ satisfying \eqref{eq:unit-product-zero} (which yields a forbidden zero), 
we evaluate its complement $S^c = W \setminus S$, which has size $s = |W| - r \ge k$. 
Letting $P_W = \prod_{w \in W} [\xi_w]$ denote the fixed total product of all elements in $W$, the condition $\prod_{a \in S} [\xi_a] = [\kappa]$
is algebraically equivalent to $\prod_{a \in S^c} [\xi_a] = P_W [\kappa]^{-1}$. 
This explicitly reduces the problem to finding $s$ distinct elements in $W$ whose product equals this prescribed target in $G$.

Assume $m > H_q$. We initially fix $s-2$ distinct factors from $W$. The pool of remaining available elements, $D$, has size $|W|-(s-2) = r+2 = m+4$. 
For the residual target $g \in G$, the pigeonhole principle bounds the number of pairs $(x,y) \in D \times D$ such that $xy = g$. Specifically:
$$
 |D\cap gD^{-1}|\ge 2|D|-|G| \ge 2(m+4)-(q+1) \ge 3.
$$
To ensure the final two factors are distinct, we must avoid solutions to $x^2=g$. 
Except in the additive group of characteristic two, the equation $x^2=g$ 
has at most two solutions in $G$. Since $|D\cap gD^{-1}| \ge 3$, 
there exists at least one pair of strictly distinct factors 
$x, y \in D$ yielding the target product.

In the exceptional case of $(\F_q,+)$ in characteristic two, we require $s\ge3$. 
Fix $s-3$ factors. There remain $|W|-(s-3)=m+5$ available
elements, so one can choose an additional element $z$ different
from the residual target $g'$. The new target $g=g'-z$ is nonzero,
and the remaining pool $D$ has size $m+4$.
Consequently, the translating involution $x \mapsto g - x$ 
has no fixed points. The same counting argument yields $|D \cap (g-D)| \ge 3 > 0$, 
guaranteeing a distinct pair. 
This successful construction of $S$ yields a forbidden zero, 
finalizing the contradiction and proving the bound.
\end{proof} 

For a linear functional $\mathcal{F}: \F_q[X]_{\le 4} \to \F_q$, 
we define its associated $3 \times 3$ 
Hankel matrix as $H(\mathcal{F}) = \big(\mathcal{F}(X^{i+j})\big)_{0 \le i, j \le 2}$. 
For a given subset $A \subseteq \F_q$, we say that $\mathcal{F}$ \emph{survives} $A$ if $\mathcal{F}$ does 
not annihilate any polynomial associated with a $4$-subset of $A$; that is, 
$\mathcal{F}(p_S) \neq 0$ for all $S \in \binom{A}{4}$.
Let $[\mathcal{F}]$ denote the projective equivalence class of $\mathcal{F}$ 
under nonzero scalar multiplication. 
We define the set of non-degenerate, surviving functional classes on $A$ as
$$
 \mathcal B(A) = \left\{ [\mathcal F] \in \mathbb{P}\big((\F_q[X]_{\le 4})^*\big) \;\bigg|\; 
\det H(\mathcal F) \ne 0 \text{ and } \mathcal F \text{ survives } A \right\},
$$
and we let $b(A) = |\mathcal B(A)|$ denote the finite cardinality of this set.

The next observation ensures that the nonsingular contractions
obtained from different omission pairs represent distinct projective
classes.

\begin{lemma}\label{lem:collision}
Let $d+4=t$, and let distinct $d$-sets $D,D'$ differ in $j=|D\setminus D'|\in\{1,2\}$ elements. 
If their nonzero quartic contractions on $\F_q[X]_{\le 4}$, denoted $\mathcal{L}(p_D\,\cdot)$ 
and $\mathcal{L}(p_{D'}\,\cdot)$, 
are projectively equal, then their $3 \times 3$ Hankel matrix $H(\mathcal{L}(p_D\,\cdot))$ 
has rank at most $j$.
\end{lemma}
\begin{proof}
Let $S = D \cap D'$ be the common elements of the two sets. 
We define an intermediate linear functional $\mathcal{M}$ on $\F_q[X]_{\le 4+j}$:
$ \mathcal{M}(f) = \mathcal{L}(p_S \cdot f). $
The hypothesis that the quartic contractions are projectively equal means there exists a constant $\alpha \in \F_q^*$ 
such that $\mathcal{L}(p_D f) = \alpha \mathcal{L}(p_{D'} f)$ for all $f \in \F_q[X]_{\le 4}$. 
Substituting the factorizations of $p_D= p_S p_{D \setminus D'}$ and $p_{D'} = p_S p_{D' \setminus D}$, 
this equivalence can be reduced to
$$\mathcal{M}(Q \cdot f) = 0 \quad \text{for all } f \in \F_q[X]_{\le 4}, $$
where $Q(X) = p_{D \setminus D'}(X) - \alpha p_{D' \setminus D}(X)$. 
It follows that $Q(X)$ is a nonzero polynomial of degree at most $j$.
Homogenize $\mathcal M$ in total degree $4+j$ and homogenize
$Q$ to degree $j$, retaining the same notation. Then
$\mathcal M(Q\mathcal H_4)=0$, where $\mathcal H_4$ is the
five-dimensional space of binary forms of degree four.

Our objective is to bound the rank of the $3 \times 3$ Hankel matrix $H\big(\mathcal{L}(p_D\,\cdot)\big)$, 
which is exactly the Gram matrix of the symmetric bilinear pairing 
$\langle F, G \rangle = \mathcal{L}(p_D \cdot F \cdot G)= \mathcal{M}(p_{D \setminus D'} \cdot F \cdot G) $ 
on the space of quadratic forms $\mathcal{H}_2$. 
Consider the subspace $Q\mathcal{H}_{2-j}$ inside $\mathcal{H}_2$. For any 
element $Q \cdot R \in Q\mathcal{H}_{2-j}$ (where $R \in \mathcal{H}_{2-j}$) and any test form $G \in \mathcal{H}_2$, 
we have
$$ \mathcal{M}\big(Q \cdot (p_{D \setminus D'} R G)\big) \in \mathcal{M}(Q\mathcal{H}_4) = \{0\}. $$
This proves that the entire subspace $Q\mathcal{H}_{2-j}$ lies in the radical 
(the nullspace) of the bilinear pairing on $\mathcal{H}_2$. 
By the rank-nullity theorem applied to the 3-dimensional 
space $\mathcal{H}_2$, the rank of the associated Gram matrix is bounded by
$ \rank H\big(\mathcal{L}(p_D\,\cdot)\big) \le 3 - 
\dim(Q\mathcal{H}_{2-j}) \le 3 - (3-j) = j. $
\end{proof}

We can now combine the class-size bound with the collision argument
to obtain the lower bound used by the detector construction.

\begin{lemma}\label{thm:fiber-budget}
Let $q\ge16$ and $2+\iota_q\le k<\lfloor(q-1)/2\rfloor$.
Put $t=q-k-1$ and $n=t-2$. If a noncanonical $\mathcal L$
satisfies \eqref{eq:NV}, then, for every $A\subseteq\F_q$ with
$|A|\le k+2$,
\begin{equation}\label{eq:fiber-budget}
 b(A)\ge (n-H_q)H_q>0.
\end{equation}
In particular, there exists a subset $D \subseteq \F_q \setminus A$ of size $t-4$ 
such that the restricted functional 
$\mathcal F_D = \mathcal L(p_D\,\cdot)$ on $\F_q[X]_{\le 4}$ 
produces a nonsingular $3 \times 3$ Hankel matrix and survives $A$.
\end{lemma}
\begin{proof}
It follows from Lemma~\ref{lem:quadratic-seed} that there exists a subset 
$U \subseteq \F_q \setminus A$ of size $n = t-2$ 
such that the $2 \times 2$ Hankel matrix $H_{1,1}(\mathcal L(p_U \,\cdot\,))$ 
is nonsingular.
Taking $\mathcal R=\mathcal L$ in Lemma~\ref{lem:kernel-omission}, 
each $u \in U$ is associated with a one-dimensional kernel $[Q_u]$. 
Let $C_1, C_2, \ldots$ 
be the resulting equivalence classes in $U$, and let $m_i = |C_i|$ be their sizes. 
Lemma~\ref{lem:kernel-omission} establishes that
selecting $u_1$ and $u_2$ from different equivalence classes yields 
a subset $U \setminus \{u_1, u_2\}$
whose associated quartic contraction $\mathcal L(p_{U \setminus \{u_1, u_2\}} \,\cdot\,)$ 
produces a nonsingular $3 \times 3$ Hankel matrix.
The number of such pairs is exactly the number of edges in a complete multipartite graph 
with part sizes $m_i$, given by $\frac{1}{2}(n^2 - \sum_i m_i^2)$.

We bound the size of each equivalence class $C\subseteq U$.
Let $\mathcal M = \mathcal L(p_{U \setminus C} \,\cdot\,)$ 
and $V = \F_q \setminus U$. 
Note that $|V| = q - (q-k-3) = k+3$. Lemma~\ref{lem:kernel-omission} 
gives $\rank H_{2, |C|}(\mathcal M) \le 2$. 
The nonsingularity of $H_{1,1}(\mathcal L(p_U \,\cdot\,))$ 
ensures $\det H_{1,1}(\mathcal M(p_C \,\cdot\,)) \ne 0$. 
Furthermore, for any $S \in \binom{C \cup V}{|C|+2}$, 
the disjoint union $(U \setminus C) \cup S$ has size $(t - 2 - |C|) + (|C| + 2) = t$. 
The hypothesis \eqref{eq:NV} requires $\mathcal L$ 
to be non-vanishing on any $t$-subset of $\F_q$, 
guaranteeing $\mathcal M(p_S) = \mathcal L(p_{(U \setminus C) \cup S}) \ne 0$. 
The hypotheses of Lemma~\ref{lem:dense-quad} 
are therefore completely satisfied, which strictly bounds each part size by 
$m_i \le H_q = \lfloor(q-5)/2\rfloor$. (If $m_i \le 2$, 
this bound holds trivially since $q \ge 16$).

The dimension bounds $2+\iota_q \le k < \lfloor(q-1)/2\rfloor$ imply $H_q < n \le 2H_q$. 
The number of nonsingular omission pairs equals the edge count of the complete multipartite graph, 
which is $\frac{1}{2}(n^2 - \sum_i m_i^2)$. 
Minimizing this count requires maximizing the strictly convex function $\sum_i m_i^2$ 
subject to $\sum_i m_i = n$ and $m_i \le H_q$. 
For $n \in (H_q, 2H_q]$, this maximum is uniquely attained at the 
boundary partition $\{H_q, n-H_q\}$. 
Consequently, $\sum_i m_i^2 \le H_q^2 + (n-H_q)^2$, yielding at least
$$
 \frac{n^2 - \big(H_q^2 + (n-H_q)^2\big)}{2} = H_q(n-H_q)
$$
nonsingular pairs.

For each such pair $\{u_1, u_2\}$, define $D = U \setminus \{u_1, u_2\}$. 
These sets have size $t-4$ and pairwise differ by at most two elements. 
If any two such sets $D, D'$ produced projectively equivalent contractions, 
Lemma~\ref{lem:collision} would force the rank of 
their Hankel matrices to be at most 2, contradicting their established nonsingularity. 
Thus, these $H_q(n-H_q)$ projective classes are strictly distinct. 
Finally, for any $S \in \binom{A}{4}$, $D \cap S = \varnothing$, 
making $D \cup S$ a $t$-set. Condition \eqref{eq:NV} 
implies $\mathcal F_D(p_S) = \mathcal L(p_{D \cup S}) \ne 0$, 
verifying that every such contraction survives $A$.
\end{proof}

\section{Involution detectors}\label{sec:graphs}
The multiplicity bound gives a lower estimate for the surviving
quartic functionals associated with a noncanonical deep hole. We
construct small sets with fewer surviving functionals by combining
projective involutions with independent-set estimates, and then
propagate the resulting exclusion through the intermediate dimensions.

Let $\mathcal F: \F_q[X]_{\le 4} \to \F_q$ be a linear functional whose $3 \times 3$ 
Hankel matrix $H(\mathcal F)$ is nonsingular. 
Let $T \subseteq \F_q$ be a subset of size $|T| = s \ge 4$. For each 2-subset $C \in \binom{T}{2}$, 
define the parameters:
\[
 A_C=\mathcal F(p_C),\quad B_C=\mathcal F(Xp_C),\quad
 E_C=\mathcal F(X^2p_C),\quad \Delta_C=B_C^2-A_CE_C,\quad
 J_C=\begin{pmatrix}B_C&-E_C\\ A_C&-B_C\end{pmatrix}.
\]
We call a 2-subset $C$ \emph{usable} when its associated matrix $J_C$ 
is both invertible  and nonscalar.

\begin{lemma}\label{lem:anchors1}
Among all pairs in $T$, at least $\frac{s(s-3)}{2}-\iota_q$ of the 2-subsets are usable. 
Furthermore, if $\mathcal F$ survives $T$, 
then the singular 2-subsets, namely the pairs in
\[
 \left\{ C \in \binom{T}{2} \;\middle|\; \Delta_C = 0 \right\},
\]
are mutually disjoint. Consequently, 
this guarantees that at least $\binom{s}{2}-\lfloor s/2 \rfloor-\iota_q$ of the 2-subsets are usable.
\end{lemma}

\begin{proof}
Fix $z \in T$. Let $\mathbf{r}_0, \mathbf{r}_1, \mathbf{r}_2$ 
be the rows of the nonsingular 
Hankel matrix $H(\mathcal F)$,
and thus the two shifted rows $\mathbf{r}_1 - z\mathbf{r}_0$ and $\mathbf{r}_2 - z\mathbf{r}_1$ 
are linearly independent. 
Let $\mathbf{v}_0, \mathbf{v}_1, \mathbf{v}_2 \in \F_q^2$ 
denote the columns of 
the $2 \times 3$ matrix formed by these two rows. 
Its one-dimensional right nullspace is spanned by 
the coefficients of the non-zero polynomial
\[
Q_z(X) = \det(\mathbf{v}_1, \mathbf{v}_2) - 
X\det(\mathbf{v}_0, \mathbf{v}_2) + X^2\det(\mathbf{v}_0, \mathbf{v}_1).
\]
For a 2-subset $C = \{z, w\}$, direct expansion of 
the functional $\mathcal F$ over $p_C(X) = (X-z)(X-w)$ 
reveals that 
$\mathbf{v}_1 - w \mathbf{v}_0 = (A_C, B_C)^{\mathsf T}$ and $\mathbf{v}_2 - 
w \mathbf{v}_1 = (B_C, E_C)^{\mathsf T}$. 
Therefore,
$
 Q_z(w) = \det(\mathbf{v}_1 - w \mathbf{v}_0, \mathbf{v}_2-
 w \mathbf{v}_1)=-\Delta_C.
$
Thus, a 2-subset $C = \{z, w\}$ is singular (i.e., $\Delta_C = 0$) 
if and only if $w$ is a root of $Q_z(X)$,
which implies that  there are at most $s$ singular 2-subsets in $T$. 

Now, suppose $\mathcal F$ survives $T$. 
Assume for contradiction that there exists an element $z \in T$ and two distinct elements 
$w_1, w_2 \in T \setminus \{z\}$ such that $\Delta_{\{z,w_1\}} = \Delta_{\{z,w_2\}} = 0$.
We must have $Q_z(X) = \lambda(X-w_1)(X-w_2)$ for some non-zero scalar $\lambda$. 
By the definition of $Q_z(X)$, we have two polynomial kernel equations:
\[
 \mathcal F\big((X-z)Q_z(X)\big) = 0 \quad \text{and} \quad \mathcal F\big(X(X-z)Q_z(X)\big) = 0.
\]
Since $|T| = s \ge 4$, we can arbitrarily choose a fourth point $w \in T \setminus \{z, w_1, w_2\}$. 
The root polynomial of the $4$-subset $S = \{z, w_1, w_2, w\}$ 
can be factored as $p_S(X) = (X-w)(X-z) \cdot \lambda^{-1} Q_z(X)$. Expanding this algebraically 
and applying the functional $\mathcal F$ yields:
\[
 \mathcal F(p_S) = \lambda^{-1} \Big[ \mathcal F\big(X(X-z)Q_z(X)\big) 
-w\mathcal F\big((X-z)Q_z(X)\big) \Big].
\]
Substituting the two kernel equations into this expression gives $\mathcal F(p_S) = 0$, 
which contradicts the assumption that $\mathcal F$ survives $T$. 
Thus, every element in $T$ can have at most one singular partner, 
proving that the singular 2-subsets are mutually disjoint.

For $\Delta_C\ne0$, the matrix 
$J_C=\left(\begin{smallmatrix}B_C&-E_C\\A_C&-B_C\end{smallmatrix}\right)$ 
acts on the projective line $\mathbb{P}^1(\F_q)$ 
via the fractional linear transformation $y = (B_Cx-E_C)/(A_Cx-B_C)$. 
Direct matrix multiplication shows $J_C^2=\Delta_C I$. 
Thus its projective order divides two. It has order exactly two
when $J_C$ is nonscalar, that is, when $C$ is usable.
A scalar matrix $J_C=\lambda I$ forces $A_C=E_C=0$ and $B_C\ne0$; 
this degeneracy only occurs in characteristic two ($\iota_q=1$). 
The coefficients of the monic quadratic $p_C$ satisfy the exact identity 
$H(\mathcal F)\operatorname{coeff}(p_C)=(A_C,B_C,E_C)^{\mathsf T}$. 
Since $H(\mathcal F)$ is nonsingular, a scalar matrix $J_C$ implies 
$\operatorname{coeff}(p_C)=B_C H(\mathcal F)^{-1}(0,1,0)^{\mathsf T}$. 
The requirement that the quadratic polynomial $p_C$ is monic uniquely fixes 
the scalar $B_C$, which completely determines the polynomial $p_C(X)$. 
Hence, at most one such scalar 2-subset $C$ can exist.

Similarly, if two distinct usable 2-subsets $C_1$ and $C_2$ define identical projective maps, 
their projective matrices must be proportional, 
meaning $(A_{C_1}, B_{C_1}, E_{C_1})^{\mathsf T} = \lambda (A_{C_2}, B_{C_2}, E_{C_2})^{\mathsf T}$ 
for some $\lambda \in \F_q^*$. 
Applying $H(\mathcal F)^{-1}$ yields 
$\operatorname{coeff}(p_{C_1}) = \lambda \operatorname{coeff}(p_{C_2})$. The monic 
requirement then forces $\lambda = 1$, 
resulting in identical polynomials $p_{C_1} = p_{C_2}$ and contradicting $C_1 \neq C_2$.

Subtracting at most $s$ singular 2-subsets and $\iota_q$ 
scalar 2-subsets from $\binom{s}{2}$ yields at least 
$\frac{s(s-3)}{2}-\iota_q$ usable 2-subsets. If $\mathcal F$ survives $T$, 
this disjointness strictly caps the number of singular 2-subsets at 
$\lfloor s/2 \rfloor$, raising the usable lower bound to 
$\binom{s}{2}-\lfloor s/2 \rfloor-\iota_q$.
\end{proof}

We next combine the two-element orbits of these involutions. For a
nonempty family $\mathcal C$ of usable pairs, let $G_{\mathcal F}$
be the graph on $\F_q\setminus T$ whose edges are the distinct sets
$\{x,J_C(x)\}$ with $C\in\mathcal C$, $x\ne J_C(x)$, and both
endpoints outside $T$.

\begin{lemma}\label{lem:anchors2}
The graph $G_{\mathcal F}$ satisfies
\begin{equation}\label{eq:graph-counts}
 \Delta(G_{\mathcal F})\le |\mathcal C|,\qquad
 e(G_{\mathcal F})\ge
 \frac{|\mathcal C|}{2}(q-2s-|\mathcal C|-\delta_q).
\end{equation}
\end{lemma}
\begin{proof}
Each involution supplies a matching, so the degree bound is immediate.
For a usable pair $C$, the fixed points of $J_C$ are its projective
eigenvectors. Since $J_C^2=\Delta_C I$, there are at most two in
odd characteristic and exactly one in characteristic two. Hence
$J_C$ has at least $(q+1-\delta_q)/2$ two-element orbits on
$\PP^1(\F_q)$. Deleting the $s+1$ points of $T\cup\{\infty\}$
leaves at least $(q-2s-\delta_q-1)/2$ edges.

By Lemma~\ref{lem:anchors1}, distinct usable pairs give distinct
projective transformations. Two such transformations share at most
one external edge: two common two-element orbits would give four
distinct points on which the transformations agree, whereas a
projective transformation is determined by its values at three
distinct points. Inclusion--exclusion therefore gives
\begin{align*}
 e(G_{\mathcal F})
 &\ge |\mathcal C|\frac{q-2s-\delta_q-1}{2}
             -\binom{|\mathcal C|}{2}\\
 &=\frac{|\mathcal C|}{2}(q-2s-|\mathcal C|-\delta_q).
\end{align*}
\end{proof}

The edge count becomes useful through an upper bound for the
independent-set polynomial.

\begin{lemma}\label{lem:degree-bound}
	Let $G$ be a graph with $n$ vertices, $e$ edges, and maximum degree at most $m \ge 1$. 
	For any real number $\lambda > 0$, the independent set polynomial of $G$, 
	defined as $P_G(\lambda)= \sum_{I \in \mathcal{I}(G)} \lambda^{|I|}$ 
	where $\mathcal{I}(G)$ denotes the collection of all independent sets in $G$, satisfies
	\begin{equation}\label{eq:weighted-degree}
		P_G(\lambda)\le(1+\lambda)^n
		\bigl(2(1+\lambda)^{-m}-(1+\lambda)^{-2m}\bigr)^{e/m^2}.
	\end{equation}
\end{lemma}
\begin{proof}
By Lemma~\ref{lem:sssz}, 
$P_G(\lambda) \le (1+\lambda)^{n_{iso}} \prod_{uv \in E} 
\big((1+\lambda)^{d_u}+(1+\lambda)^{d_v}-1\big)^{1/(d_u d_v)}$, 
where $n_{iso}$ is the number of isolated vertices. 
To normalize this bound, we divide by $(1+\lambda)^n$.
Since the total number of vertices is 
$n = n_{iso} + \sum_{uv \in E} \left(\frac{1}{d_u} + \frac{1}{d_v}\right)$, 
the normalized bound simplifies to $\prod_{uv \in E} f(d_u, d_v)$, 
where we substitute $w=(1+\lambda)^{-1} \in (0,1)$ and define the edge 
contribution factor as $f(a,b) = (w^a+w^b-w^{a+b})^{1/(ab)}$.

To prove that $f(a,b)$ monotonically increases with the vertex degrees, 
we analyze its natural logarithm: $\ln f(a,b) = -\frac{g(a)}{ab}$, 
where $g(a) = -\ln(w^a+w^b-w^{a+b})$. 
For a fixed $b>0$, let $c = 1-w^b \in (0,1)$ and $L = -\ln w > 0$. 
We can rewrite $g(a)$ as $g(a) = -\ln(1-c+c e^{-La})$. 
Taking derivatives with respect to $a$ yields:
$$g'(a) = \frac{cL}{(1-c)e^{La}+c},
\quad g''(a) = -\frac{c(1-c)L^2 e^{La}}{((1-c)e^{La}+c)^2} < 0.$$
This shows that $g(a)$ is strictly concave. Since $g(0) = -\ln(1) = 0$, 
the properties of concave functions dictate that the ratio $g(a)/a$ 
is strictly decreasing for $a > 0$. 
Consequently, the term $\ln f(a,b) = -\frac{1}{b} \frac{g(a)}{a}$ is 
strictly increasing with respect to $a$ (and, by symmetry, with respect to $b$). 
We can therefore upper bound the product by replacing every endpoint degree with the 
graph's maximum degree $m \ge 1$. This yields $f(a,b) \le f(m,m) = (2w^m - w^{2m})^{1/m^2}$. 
Taking the product over all $e$ edges and multiplying the initial $(1+\lambda)^n$ factor 
back to the right-hand side completes the proof, with isolated vertices natively 
accounted for by the initial factorization.
\end{proof}

To propagate a detector to higher dimensions, we also count exterior
points that already violate survival together with three anchor
points. For a functional $\mathcal F$ surviving $T$, put
\[
 N_{\mathcal F}=\{x\in\F_q\setminus T:
 \mathcal F(p_{U\cup\{x\}})=0
 \text{ for some }U\in\tbinom T3\},\qquad
 \nu=|N_{\mathcal F}|.
\]

\begin{lemma}\label{thm:joint-deletion}
Let $q\ge16$ and $|T|=s\ge6$. If a nonsingular functional
$\mathcal F$ survives $T$, then
$\nu\ge\max\{0,s-5-\delta_q\}$.
\end{lemma}
\begin{proof}
Lemma~\ref{lem:anchors1} supplies a usable pair $C\subseteq T$.
Consider the images under $J_C$ of the $s-2$ points of
$T\setminus C$. At most $\delta_q$ points are fixed, and at most
three points have their image in $C\cup\{\infty\}$. No remaining
point can map to another point of $T\setminus C$, since such a
two-element orbit $\{u,v\}$ would give
$\mathcal F(p_{C\cup\{u,v\}})=0$, contradicting survival on $T$.
Thus at least $s-5-\delta_q$ distinct points $u\in T\setminus C$
map to distinct exterior points $x=J_C(u)$. Each satisfies
$\mathcal F(p_{C\cup\{u,x\}})=0$, so belongs to
$N_{\mathcal F}$. Nonnegativity completes the bound.
\end{proof}

The preceding estimate allows a small detector to be enlarged one
point at a time, with a decrease matching the multiplicity bound.

\begin{lemma}\label{thm:detector-transport}
	Let $q\ge31$ and $8\le k_0<\lfloor(q-1)/2\rfloor$. If there exists a subset $A_0 \subseteq \F_q$ with $|A_0| \le k_0+2$ satisfying
	\begin{equation}\label{eq:seed-detector}
		b(A_0) < \left\lfloor\frac{q-5}{2}\right\rfloor \left(\lfloor q/2 \rfloor - k_0\right),
	\end{equation}
	then every deep hole is canonical in every dimension $k$ in the range
	\begin{equation}\label{eq:detector-tail}
		k_0\le k<\lfloor(q-1)/2\rfloor.
	\end{equation}
For the existence of such $A_0$, it is sufficient to show that for some integer $c \le k_0+2$, the uniform expectation over $c$-subsets satisfies
	\[
	\mathbb{E}_{A \in \binom{\F_q}{c}}[b(A)] < \left\lfloor\frac{q-5}{2}\right\rfloor \left(\lfloor q/2 \rfloor - k_0\right).
	\]
\end{lemma}
\begin{proof}
Enlarge $A_0$ to size $k_0+2$. For $A \supseteq A_0$ such that $6 \le |A| = a < q$, 
and for each $[\mathcal{F}] \in \mathcal{B}(A)$, 
let $N_{\mathcal{F}} = \{x \in \mathbb{F}_q \setminus A : \text{there is }
U \in \binom{A}{3} \text{ such that } \mathcal{F}(p_{U \cup \{x\}}) = 0\}$ denote the nullifying set. 
By Lemma~\ref{thm:joint-deletion}, $\nu = |N_{\mathcal{F}}| \ge a - 5 - \delta_q$.
Since the class $[\mathcal{F}]$ belongs to $\mathcal{B}(A \cup \{x\})$ if and only if 
$[\mathcal{F}] \in \mathcal{B}(A)$ and $x \notin N_{\mathcal{F}}$, we can use the indicator function 
$\mathbf{1}_{\{x \notin N_{\mathcal{F}}\}}$ to exchange the order of summation:
\begin{align*}
	\sum_{x \in \mathbb{F}_q \setminus A} b(A \cup \{x\}) 
	&= \sum_{x \in \mathbb{F}_q \setminus A} 
	\sum_{[\mathcal{F}] \in \mathcal{B}(A)} \mathbf{1}_{\{x \notin N_{\mathcal{F}}\}} \\
	&= \sum_{[\mathcal{F}] \in \mathcal{B}(A)} (q - a - |N_{\mathcal{F}}|) \\
	&\le \sum_{[\mathcal{F}] \in \mathcal{B}(A)} \big( q - a - (a - 5 - \delta_q) \big) \\
	&= b(A) \big( q - a - (a - 5 - \delta_q) \big).
\end{align*}

Dividing both sides by $q-a$ yields:

\begin{equation}\label{eq:survivor-growth}
	\frac{1}{q-a}\sum_{x \notin A}b(A \cup \{x\}) \le b(A)\left(1-\frac{a-5-\delta_q}{q-a}\right).
\end{equation}

We prove by induction on $j$ that if there exists an initial set $A_0$ of size $k_0+2$ 
with $b(A_0) < Q_{\rm f}(q, k_0)=\lfloor\frac{q-5}{2}\rfloor \left(\lfloor q/2 \rfloor - k_0\right)$, 
then for all $j$ in the range $k_0 \le j < \lfloor(q-1)/2\rfloor$, there exists a set $A_j$ of size $j+2$ 
such that $b(A_j) < Q_{\rm f}(q, j)$. This strict upper bound contradicts the necessary condition 
$b(A) \ge Q_{\rm f}(q, j)$ for noncanonical deep holes (Lemma~\ref{thm:fiber-budget}), 
thereby proving that every deep hole in this range is canonical.

For the base case, by hypothesis, there exists a set $A_0$ of size $k_0+2$ satisfying $b(A_0) < Q_{\rm f}(q, k_0)$.

For the inductive step, assume that for some integer $j \ge k_0$, there exists a set $A_j$ of size $j+2$ such that $b(A_j) < Q_{\rm f}(q, j)$.

For odd $q \ge 31$  and $8 \le j \le \frac{q-5}{2}$:
$$
(j - 5)\left(\frac{q-1}{2} - j\right) - (q - j - 2) = (j-8)\left(\frac{q-5}{2} - j\right) + \frac{q-31}{2} \ge 0.
$$

For even $q \ge 32$  and $8 \le j \le \frac{q}{2} - 3$:
$$
(j - 4)\left(\frac{q}{2} - j\right) - (q - j - 2) = (j-8)\left(\frac{q}{2} - 3 - j\right) + q - 22 > 0.
$$

Thus, $(j - 3 - \delta_q)(\lfloor q/2 \rfloor - j) - (q - j - 2) \ge 0$ holds universally. Rearranging this inequality yields:
\begin{equation}\label{eq:transport-ratio}
	1-\frac{j-3-\delta_q}{q-j-2} \le \frac{\lfloor q/2 \rfloor - j - 1}{\lfloor q/2 \rfloor - j} = \frac{Q_{\rm f}(q, j+1)}{Q_{\rm f}(q, j)}.
\end{equation}

By Equation \ref{eq:survivor-growth}, the average number of surviving functional classes upon 
adding a single exterior point $x \notin A_j$ is bounded by:
\begin{equation*}
	\frac{1}{q-j-2}\sum_{x\notin A_j}b(A_j\cup\{x\}) \le b(A_j)\left(1-\frac{j-3-\delta_q}{q-j-2}\right).
\end{equation*}

Since the minimum value over a finite set is always less than or equal to its average, there exists at least one element $x^* \notin A_j$ such that the new set $A_{j+1} = A_j \cup \{x^*\}$ satisfies:
\begin{equation*}
	b(A_{j+1}) \le b(A_j)\left(1-\frac{j-3-\delta_q}{q-j-2}\right).
\end{equation*}

Applying the inductive hypothesis $b(A_j) < Q_{\rm f}(q, j)$ and the transport ratio bound \eqref{eq:transport-ratio}, we obtain:
\begin{equation*}
	b(A_{j+1}) < Q_{\rm f}(q, j) \left( \frac{Q_{\rm f}(q, j+1)}{Q_{\rm f}(q, j)} \right) = Q_{\rm f}(q, j+1).
\end{equation*}

This completes the induction.
Since $\mathbb{E}[b(A)] < Q_{\rm f}(q, k_0)$ for $|A|=c \le k_0+2$, 
the probabilistic method guarantees a $c$-subset $A^*$ exists with $b(A^*) < Q_{\rm f}(q, k_0)$. 
Because adding points to the evaluation set imposes additional non-vanishing constraints, 
$b(\cdot)$ is monotonically non-increasing under supersets. 
Arbitrarily padding $A^*$ to size $k_0+2$ to form $A_0$ ensures $b(A_0) \le b(A^*) < Q_{\rm f}(q, k_0)$.
\end{proof}

We can improve the expected number of survivors by counting only
the functionals that survive the anchor set itself. The following
second-moment estimate provides that count.

\begin{lemma}\label{lem:anchor-survivors}
	Let $T \subseteq \F_q$ with $|T| = s \ge 4$. Up to projective equivalence, 
	the number of linear functionals $\mathcal{F}: \F_q[X]_{\le 4} \to \F_q$ satisfying $\mathcal{F}(p_S) \ne 0$ 
	for all $S \in \binom{T}{4}$ is at most
	\begin{equation}\label{eq:survivors}
		N_s(q) = \frac{q^5}{q-1+\binom s4}.
	\end{equation}
	This bound holds whether $\det H(\mathcal{F}) \ne 0$ is imposed or not.
\end{lemma}
\begin{proof}
	Let $\mathcal{F}$ be drawn uniformly at random from the 5-dimensional dual space $(\F_q[X]_{\le 4})^*$. Define the random variable $Z = \sum_{S \in \binom{T}{4}} \mathbf{1}_{\{\mathcal{F}(p_S) = 0\}}$ to count the number of 4-subsets zeroed by $\mathcal{F}$.
	
	Because distinct 4-subsets produce distinct monic polynomials of degree 4, any pair of such polynomials is linearly independent. Consequently, their zero indicators are pairwise independent, and each evaluates to zero with probability $1/q$. 
	
	By linearity of expectation, the first moment is:
	\[
	\mathbb{E}[Z] = \binom{s}{4} \frac{1}{q}.
	\]
	Using pairwise independence, the second moment expands as:
	\[
	\mathbb{E}[Z^2] = \binom{s}{4} \frac{1}{q} + \binom{s}{4}\left(\binom{s}{4}-1\right) \frac{1}{q^2} = \frac{\binom{s}{4} \left( q - 1 + \binom{s}{4} \right)}{q^2}.
	\]
	
	Applying the Cauchy--Schwarz inequality ($\mathbb{P}(Z>0) \ge (\mathbb{E}[Z])^2 / \mathbb{E}[Z^2]$) yields:
	\[
	\mathbb{P}(Z>0) \ge \frac{\binom{s}{4}}{q - 1 + \binom{s}{4}}.
	\]
	Therefore, the probability that a functional survives on all 4-subsets is:
	\[
	\mathbb{P}(Z=0) = 1 - \mathbb{P}(Z>0) \le \frac{q - 1}{q - 1 + \binom{s}{4}}.
	\]
	
	Multiplying by the $q^5$ total vectors in the dual space shows there are at most 
	$q^5 \frac{q - 1}{q - 1 + \binom{s}{4}}$ surviving vectors. Since $s \ge 4$, 
	the identically zero functional has $Z \ge 1$, which guarantees that every 
	surviving vector is non-zero. 
	Dividing by $q-1$ to account for projective scalar equivalence proves the bound.
\end{proof}

We first record the independent-set estimate in the notation used
for the analytic comparison. Put
\begin{equation}\label{eq:parameters}
 N(q)=q^4+q^3+q^2+q+1,\quad c=\frac{q+M+2}{2},\quad
 \beta=\frac{q-2s-M-2}{2M},
\end{equation}
and, when $q>2s+M+2$ and $0<r<c$, define
\begin{equation}\label{eq:Phi}
 \Phi(q;s,r,M)=N(q)2^\beta
 \frac{c^c}{\binom{q-s}{r}r^r(c-r)^{c-r}}.
\end{equation}
\begin{lemma}\label{thm:detector}
Let $s\ge6$, $1\le M\le\binom{s-3}{2}$, $q>2s+M+2$, and
$1\le r\le q-s$ with $r<c$. For every fixed $s$-set $T$ and a uniform
$r$-set $R\subseteq\F_q\setminus T$,
\begin{equation}\label{eq:expected-detector}
 \mathbb E b(T\cup R)\le\Phi(q;s,r,M).
\end{equation}
\end{lemma}
\begin{proof}
Lemmas~\ref{lem:anchors1} and~\ref{lem:anchors2} supply,
by choosing $M$ usable pairs for each nonsingular $\mathcal F$,
a graph with $n=q-s$ vertices, maximum degree at most $M$, and at least
$e_0=M(q-2s-M-2)/2$ edges. Put $w=(1+\lambda)^{-1}$ and
$\theta=2w^M-w^{2M}$. Since $0<\theta<1$,
Lemma~\ref{lem:degree-bound} gives
$$
 P_{G_{\mathcal F}}(\lambda)
 \le(1+\lambda)^n\theta^{e_0/M^2}
 \le 2^{e_0/M^2}(1+\lambda)^{n-e_0/M}
 =2^\beta(1+\lambda)^c.
$$
Survival requires $R$ to be independent. Bound its probability by
$2^\beta(1+\lambda)^c/(\binom nr\lambda^r)$, choose
$\lambda=r/(c-r)$, and sum over at most $N(q)$ projective classes.
\end{proof}

For a surviving anchor set, Lemma~\ref{lem:anchors1} allows more
pairs than the preceding estimate needs. Write
\[
 M_q^+(s)=\binom s2-\lfloor s/2\rfloor-\iota_q.
\]
Combining this observation with the anchor count gives the estimate
used in the proof of the main theorem.

\begin{lemma}\label{lem:anchored-Phi-v30}
Let $q\ge31$, $s\ge6$, $1\le M\le M_q^+(s)$,
$q>2s+M+2$, and $1\le r\le q-s$ with $r<(q+M+2)/2$.
Put
\[
 N_s(q)=\frac{q^5}{q-1+\binom s4},\qquad
 \Psi(q;s,r,M)=\frac{N_s(q)}{N(q)}\Phi(q;s,r,M).
\]
For a fixed $s$-set $T$ and a uniformly chosen $r$-set
$R\subseteq\F_q\setminus T$,
$\mathbb E b(T\cup R)\le\Psi(q;s,r,M)$.
Consequently, if $8\le k_0<\lfloor(q-1)/2\rfloor$,
$s+r\le k_0+2$, and $\Psi<Q_{\rm f}(q,k_0)$, then every deep
word is canonical for $k_0\le k<\lfloor(q-1)/2\rfloor$.
\end{lemma}
\begin{proof}
Only classes surviving $T$ can contribute to $b(T\cup R)$.
There are at most $N_s(q)$ such classes by
Lemma~\ref{lem:anchor-survivors}. For each nonsingular survivor,
Lemma~\ref{lem:anchors1} supplies at least $M_q^+(s)$ usable pairs.
Choose $M$ of them. Lemma~\ref{lem:anchors2} gives a graph of
maximum degree at most $M$ and at least
$M(q-2s-M-2)/2$ edges. The proof of
Lemma~\ref{thm:detector} now applies with $N_s(q)$ in place of
$N(q)$. The last assertion follows from
Lemma~\ref{thm:detector-transport}.
\end{proof}

The last estimate reduces the construction of a detector to a
single inequality. We shall use the following logarithmic upper
bound when choosing its parameters.

\begin{lemma}\label{lem:crude}
If $q>2s+M+2$, $1\le r\le q-s$, and $r<c$, then
\begin{equation}\label{eq:crude}
 \log\Phi\le\log N(q)+\beta\log2+1+\tfrac12\log r
                +r\log\frac{c}{q-s-r+1}.
\end{equation}
\end{lemma}
\begin{proof}
Put $n=q-s$. Concavity of $\log$ gives
$\log(r!)\le\int_1^r\log x\,dx+\tfrac12\log r$.
Substitute $r!\le e\sqrt r(r/e)^r$,
$\binom nr\ge(n-r+1)^r/r!$, and
$(1+r/(c-r))^{c-r}\le e^r$ in \eqref{eq:Phi}.
\end{proof}

\section{Planar projections and completion of the proof}\label{sec:geometry}
We first obtain the lower-dimensional classification from prescribed
planar projections. We then compare it with the detector estimates
and complete the proof of Theorem~\ref{thm:main}.
We identify evaluation vectors with their representatives
of degree at most $q-1$ and multiply them coordinatewise. Put
\begin{equation}\label{eq:slope-space}
\mathcal W_q=\{v\in\F_q[X]_{\le q-1}:
2(v(X)-v(Y))=(X-Y)(v'(X)+v'(Y))\},
\end{equation}
where the equality is a polynomial identity.

The polynomial identity has the following simple solution space.

\begin{lemma}\label{lem:slope-kernel}
For $q\ge3$,
\begin{equation}\label{eq:binary-space}
\mathcal W_q=
\begin{cases}
\mathcal P_2,&q\text{ odd},\\
\operatorname{span}\bigl(\{1,X\}\cup\{X^{2j}:1\le j\le q/2-1\}\bigr),
 &q\text{ even}.
\end{cases}
\end{equation}
\end{lemma}
\begin{proof}
For $n\ge3$, the coefficients of $X^{n-1}Y$ and $X^n$ give
$nv_n=(2-n)v_n=0$. This excludes $n\ge3$ in odd characteristic.
In characteristic two, \eqref{eq:slope-space} says that $v'$ is constant.
Every quadratic satisfies the identity.
\end{proof}

Prescribing the projection center keeps track of the interpolation polynomial that must be recovered after conic completion.

\begin{lemma}\label{lem:single-conic}
Let $u$ be deep, $2\le k\le q-2$, and $E\subseteq\F_q$ have
size $c=k-2$. Let $f_E$ interpolate $u$ on $E$ with $\deg f_E<c$,
using $f_E=0$ if $c=0$. Then
\begin{equation}\label{eq:prescribed-projection}
A_E=\{[0:0:1]\}\cup
\left\{\left[1:x:\frac{u(x)-f_E(x)}{p_E(x)}\right]:x\notin E\right\}
\end{equation}
is an arc of size $q-k+3$. If $A_E$ is contained in a conic, then
$u=f_E+p_Ev$ for a quadratic $v$, and $u$ is canonical.
\end{lemma}
\begin{proof}
Three collinear affine points would make $f_E+p_Eh$, for some
$\deg h\le1$, agree with $u$ at $k+1$ points. Lines through
$[0:0:1]$ contain at most one affine point of $A_E$.
A conic containing these at least five arc points is nonsingular and has
an equation $Q(X,Y)+(dX+eY)Z=0$, with $Q$ quadratic and $(d,e)\ne(0,0)$.
Thus, for $w=u-f_E$,
$$(d+ex)w(x)+p_E(x)Q(1,x)=0\qquad(x\notin E).$$
If $e\ne0$, its root $z=-d/e$ must belong to $E$: otherwise
$Q(1,z)=0$ makes the conic reducible. But then
$-p_EQ(1,X)/(d+eX)$ has degree at most $k-1$ and agrees with $w$
except possibly at $z$, contradicting deepness since $q-k\ge2$.
Hence $e=0$, and $u=f_E-p_EQ(1,X)/d$ on all of $\F_q$.
This polynomial has degree at most $k$, and deepness excludes degree less
than $k$.
\end{proof}

In characteristic two, completion gives a hyperoval instead of a conic. Its affine part still satisfies the same polynomial identity.

\begin{lemma}\label{lem:hyperoval-slope}
For even $q\ge4$, a hyperoval through $[0:0:1]$ has the form
$$\{[0:0:1],[0:1:\alpha]\}\cup
\{[1:x:v(x)]:x\in\F_q\},\qquad v'\equiv\alpha.$$
In particular, $v\in\mathcal W_q$.
\end{lemma}
\begin{proof}
Every line through $[0:0:1]$ contains exactly one other hyperoval point.
For fixed $a$, the slopes $(v(x)-v(a))/(x-a)$, $x\ne a$, therefore
run through $\F_q\setminus\{\alpha\}$ and sum to $\alpha$.
Expanding in $x-a$ and summing powers gives $v'(a)=\alpha$;
all positive powers involved have degree at most $q-2$ and sum to zero.
The degree bound on $v'$ gives the polynomial identity.
\end{proof}

The common constraint from several such projections forces the original word to have small degree.

\begin{lemma}
\label{thm:binary-intersection}\label{thm:common-intersection}
For $q\ge8$, $1\le c\le\lfloor q/2\rfloor-2$, and
$B\subseteq\F_q$ of size $c+1$, put $E_b=B\setminus\{b\}$. Then
\begin{equation}\label{eq:binary-intersection}
\bigcap_{b\in B}(\mathcal P_{c-1}+p_{E_b}\mathcal W_q)=\mathcal P_{c+2}.
\end{equation}
\end{lemma}
\begin{proof}
For odd $q$, each space equals $\mathcal P_{c+2}$ by division by $p_{E_b}$.
For even $q$, set $D=q-2c-4$ and
$\mathcal T=\operatorname{span}\{X^{2j}:0\le2j\le D\}$.
If $u=f+pv$, where $p=p_{E_b}$, $\deg f<c$, and $v\in\mathcal W_q$, then
\begin{equation}\label{eq:binary-orthogonality}
\sum_{x\in\F_q}u(x)p(x)h(x)=0\qquad(h\in\mathcal T).
\end{equation}
Indeed, for $h=X^{2j}$ the term $fph$ has degree at most $q-5$;
$p^2h$ times the even part of $v$ has even exponents at most $2q-6$;
and its product with the linear term has degree at most $q-3$.
None contributes to the power sum, since
$$\sum_{x\in\F_q}x^i=
\begin{cases}-1,&i>0\text{ and }q-1\mid i,\\0,&\text{otherwise}.\end{cases}$$
Evaluation on $B$ shows that the $p_{E_b}$ form a basis of $\mathcal P_c$.
Since $c\ge1$, the intervals $[2j,2j+c]$ cover $[0,q-c-4]$.
Thus a vector in the intersection is orthogonal to $\mathcal P_{q-c-4}$.
For $u=\sum_{i=0}^{q-1}u_iX^i$ this gives
$0=\sum_xu(x)x^j=-u_{q-1-j}$ for $0\le j\le q-c-4$, hence
$\deg u\le c+2$. The reverse inclusion follows by polynomial division
and $\mathcal P_2\subseteq\mathcal W_q$.
\end{proof}

Applying this intersection identity to the prescribed projections transfers completion back to the word.

\begin{lemma}
\label{thm:hyperoval-transfer}\label{thm:common-transfer}
Let $q\ge8$, $3\le k\le\lfloor q/2\rfloor$, and let $u$ be deep.
Fix a $(k-1)$-subset $B$. If every projection $A_{B\setminus\{b\}}$,
$b\in B$, is contained in a conic for odd $q$, or a hyperoval for even
$q$, then $u$ is canonical.
\end{lemma}
\begin{proof}
The preceding completion lemmas give
$u\in\mathcal P_{k-3}+p_{B\setminus\{b\}}\mathcal W_q$ for every $b$.
Lemma~\ref{thm:common-intersection} gives $\deg u\le k$, and deepness gives equality.
\end{proof}

For odd $q$, Lemma~\ref{lem:single-conic} requires only one completed
projection, with no restriction $k\le q/2$.

In even characteristic, the completion bound and the common
intersection identity give the following range.

\begin{lemma}\label{prop:even-boundary}
For $q=2^m\ge8$, every deep word is canonical for $3\le k\le K_2(q)$,
where
\begin{equation}\label{eq:even-new-endpoint}
K_2(q)=\begin{cases}\sqrt q+1,&m\text{ even},\\
\sqrt{2q},&m\text{ odd}.\end{cases}
\end{equation}
More generally, within $3\le k\le q/2$, it suffices that
\begin{equation}\label{eq:even-new-range}
k<\sqrt q+2,\qquad\text{or}\qquad
q\text{ is nonsquare and }k<\sqrt{2q}+1.
\end{equation}
\end{lemma}
\begin{proof}
Apply Lemma~\ref{lem:hyperoval-completion} to the projections of size
$q-k+3$ and then Lemma~\ref{thm:common-transfer}. Both stated integer
endpoints are at most $q/2$.
\end{proof}

To obtain the odd nonsquare range, we thin nonconic arcs until the
two-curve bound applies. For odd characteristic $p$, write
$$g_h(z)=z^2+(3-2h)z+2-h^2-3h.$$
\begin{lemma}\label{lem:thinning-bezout}
Suppose a nonconic arc $A$ of size $q+2-\tau$ lies on a specified curve
$\phi$ of degree $s\le\tau$. Let $\sigma\ge\tau$ satisfy
$q+2-\sigma\ge6$ and $g_h(\sigma)\ge0$, where
$h=p^{\lfloor\log_p\sigma\rfloor}$.
Then $A$ has a nonconic subset $A'$ of size $q+2-\sigma$, and every such
subset lies on a curve of degree at most $\sigma+h$ coprime to $\phi$.
Consequently $q+2-\sigma\le s(\sigma+h)$.
\end{lemma}
\begin{proof}
Retain five points and a sixth outside their unique conic while thinning.
Apply Lemma~\ref{lem:two-curves} to $\phi$ padded by $\sigma-s$ linear
factors, and then apply B\'ezout.
\end{proof}

We repeatedly use the following observation. If
$h=p^{\lfloor\log_p\tau\rfloor}$ and $g_h(\tau)<0$, then $h\ge3$ and,
for $\sigma=(5h-1)/2$,
\begin{equation}\label{eq:thinning-arithmetic}
\tau<\sigma<3h\le ph,\qquad g_h(\sigma)=(h^2+12h+3)/4>0.
\end{equation}
Indeed, $g_1(z)=(z-1)(z+2)$, and $g_h$ increases on $[h,\infty)$.

For nonsquare odd fields, thinning gives the following range, which will meet the uniform detector estimate.

\begin{lemma}\label{prop:nonsquare}
Let $q=p^{2m+1}$ with odd $p$ and $m\ge1$, and put $n=p^m$. Every deep
word is canonical for $2\le k\le K_{\rm BL}(q)$, where
$$K_{\rm BL}(q)=\max\left\{n,
\left\lfloor\frac{\lfloor\sqrt{4q+4n+8}\rfloor-2n+1}{2}\right\rfloor\right\}.$$
\end{lemma}
\begin{proof}
Suppose a prescribed projection is nonconic. For $k\le n$, thin it to
size $q-n+3$, retaining six nonconic points. With $\tau_0=n-1$ and $h=n/p$,
$$g_h(\tau_0)=(p^2-2p-1)h^2+(p-1)h>0,$$
whereas
$$q-n+3-(n-1+n/p)^2
=n^2\bigl(p-(1+1/p)^2\bigr)+n(1+2/p)+2>0.$$
This contradicts Lemma~\ref{lem:two-curves} and B\'ezout.
For $k>n$, the stated bound is equivalent to
$$(\tau+n)^2+\tau-2<q,\qquad\tau=k-1.$$
It forces $p\ge5$, $n\le\tau<pn$, and $\tau<\sqrt q-n$.
Lemma~\ref{lem:two-curves} gives a contradiction unless a degree-$\tau$
containing curve exists and $g_n(\tau)<0$. In that case use
$\sigma=(5n-1)/2$ in Lemma~\ref{lem:thinning-bezout}. Its hypotheses
follow from \eqref{eq:thinning-arithmetic} and $q+2-\sigma\ge6$, but
\begin{align*}
q-\tau(\sigma+n)-\sigma+2
&>\left(p-\frac72\sqrt p+\frac72\right)n^2
 +\left(\frac{\sqrt p}{2}-3\right)n+\frac52\\
&\ge\frac7{16}n^2-2n+\frac52>0\qquad(n\ge5),
\end{align*}
again a contradiction. Apply Lemma~\ref{lem:single-conic}.
\end{proof}

Voloch's bound supplies a second range that is useful for the two cubic field orders near the final threshold.

\begin{lemma}\label{prop:voloch-range}
For $q=p^{2m+1}$ with odd $p$ and $m\ge1$, every deep word is canonical
if $2\le k\le q-2$ and
\begin{equation}\label{eq:voloch-range}
k<\frac14\sqrt{pq}-\frac{29}{16}p+2.
\end{equation}
\end{lemma}
\begin{proof}
Apply Lemmas~\ref{lem:voloch} and \ref{lem:single-conic} to a projection
of size $q-k+3$.
\end{proof}

For nonprime $q$, the preceding bounds give the canonical classification
through
\begin{equation}\label{eq:geometric-K}
\begin{aligned}
K(p^{2m})&=p^m-p^{m-1}&&(p\text{ odd}),\\
K(p^{2m+1})&=\max\left\{p^m,
\left\lfloor\frac{\lfloor\sqrt{4q+4p^m+8}\rfloor-2p^m+1}{2}\right\rfloor,
\left\lfloor\frac{4p^{m+1}-29p+31}{16}\right\rfloor\right\}
&&(p\text{ odd}),\\
K(2^m)&=K_2(2^m).
\end{aligned}
\end{equation}
For odd squares use Lemmas~\ref{lem:square} and \ref{lem:single-conic}.
The last odd term retains the strict inequality in \eqref{eq:voloch-range}.

The geometric bounds reduce the problem to the intermediate
dimensions. We first apply the detector with an initial dimension
close to four fifths of the square root of the field order.

\begin{lemma}
\label{prop:analytic-four-fifths-v30}
For every prime power $q\ge8192$, put
$k_0=\lceil4\sqrt q/5\rceil+1$.
Every deep word is canonical for
$k_0\le k<\lfloor(q-1)/2\rfloor$.
\end{lemma}
\begin{proof}
Put $x=q^{1/6}$, $y=(2q)^{1/3}$, and choose
\[
 s=\lceil y\rceil+1,\qquad
 M=M_q^+(s),\qquad a=k_0+2,\qquad r=a-s.
\]
We prove $\Psi<Q_{\rm f}(q,k_0)$ by a uniform estimate.
Since $q\ge8192$, we have $x>4489/1000>4$ and $y>25$.
The elementary bounds $1259/1000<2^{1/3}<63/50$ give
\begin{align*}
 &\frac54x^2+1\le s\le\frac{63}{50}x^2+2,\\
 &\frac{197}{250}x^4\le M
       \le\frac45x^4+\frac{63}{25}x^2+2,\\
 &\frac45x^3-\frac{63}{50}x^2+1
       \le r\le\frac45x^3-\frac54x^2+3<\frac45x^3.
\end{align*}
For the lower bound on $M$, use
$M\ge((s-1)^2-3)/2\ge(y^2-3)/2$.
These inequalities imply all the hypotheses of
Lemma~\ref{lem:anchored-Phi-v30}.
In the notation of \eqref{eq:parameters}, they also give
\[
 \beta\le\frac{127}{200}x^2-\frac12,\qquad
 Q_{\rm f}(q,k_0)\ge\frac6{25}q^2.
\]
For the latter bound, use
$H_q\ge(q-6)/2$ and
$\lfloor q/2\rfloor-k_0\ge(q-\frac85\sqrt q-5)/2$.

We retain the $q-1$ term in the denominator of $N_s(q)$.
Since $s\ge y+1$,
\begin{align*}
 q-1+\binom s4
 &\ge q-1+\frac{(y+1)y(y-1)(y-2)}{24}\\
 &\ge \frac5{48}x^8+\frac45x^6.
\end{align*}
Indeed, expansion of the first line gives
$2^{4/3}x^8/24+5x^6/6-2^{2/3}x^4/24
+2^{1/3}x^2/12-1$; use $2^{4/3}>5/2$,
$2^{2/3}<8/5$, and $x>4$.
Consequently,
\[
 \frac{N_s(q)}{Q_{\rm f}(q,k_0)}
 \le Bx^{10},\qquad
 B=\begin{cases}31,&x\le5,\\40,&x>5.\end{cases}
\]

As in Lemma~\ref{lem:crude},
\[
 \log\frac{\Psi}{Q_{\rm f}(q,k_0)}
 \le\log\frac{N_s(q)}{Q_{\rm f}(q,k_0)}
 +\beta\log2+1+\tfrac12\log r
 +r\log\frac{c}{q-s-r+1}.
\]
Since $s+r=a$, the inequality $\log(1+z)\le z$ gives
\[
 r\log\frac{c}{q-a+1}
 \le-r\log2+r\frac{M+a+1}{q-a+1}
 \le-r\log2+\frac{13}{20}x.
\]
Here the last inequality follows by multiplying the positive
denominators and using the preceding bounds on $r,M,a$.
The required numerator difference is
\[
 P(x)=\frac{x^7}{100}+\frac9{25}x^6-\frac{127}{125}x^5
      +\frac{23}{100}x^4-8x^3+\frac{119}{100}x^2
      -\frac{39}{20}x-21.
\]
Writing $z=x-4$, its coefficients in descending order are
$1/100$, $16/25$, $1373/125$, $8871/100$, $9588/25$,
$87447/100$, $88053/100$, $16892/125$, all positive.
Thus $P(x)>0$ for $x\ge4$.

It follows that
\[
 \log\frac{\Psi}{Q_{\rm f}(q,k_0)}\le-F_B(x),
\]
where
\[
 F_B(x)=\log2\left(\frac45x^3-\frac{379}{200}x^2+\frac32\right)
 -\frac{13}{20}x-1+\frac12\log\frac54
 -\frac{23}{2}\log x-\log B.
\]
The derivative satisfies $F_B'(x)>25/2$ for $x\ge4$.
Moreover,
\[
 F_{31}(4489/1000)>
 \frac{1063098440293}{5000000000000}>0,\qquad F_{40}(5)>11.
\]
These numerical constants follow from rational bounds
$\log2>693/1000$, $\log(5/4)>222/1000$,
$\log(4489/1000)<751/500$, $\log31<1717/500$,
$\log5<161/100$, and $\log40<37/10$.
For the two lower bounds use the positive series
$\log z=2\sum_{j\ge0}((z-1)/(z+1))^{2j+1}/(2j+1)$;
three terms suffice for $z=2$ and one for $z=5/4$.
For the four upper bounds, the exponential sums through degrees
$7,14,7,9$, respectively, already exceed the required arguments.
Hence $\Psi<Q_{\rm f}(q,k_0)$ in both intervals $x\le5$ and $x>5$.
Lemma~\ref{lem:anchored-Phi-v30} completes the proof.
\end{proof}
For larger fields, the same argument permits an initial dimension
below the geometric endpoints in every characteristic.

\begin{lemma}\label{prop:uniform-32768-v30}
Theorem~\ref{thm:main} holds for every prime power $q\ge2^{15}$
without finite-field searches or field-dependent scalar certificates.
\end{lemma}
\begin{proof}
Put $x=q^{1/6}$, $y=(2q)^{1/3}$, and choose
\[
 k_0=\left\lfloor\frac{11}{20}\sqrt q\right\rfloor+1,\quad
 a=k_0+2,\quad s=\lceil y\rceil+1,\quad
 M=M_q^+(s),\quad r=a-s.
\]
Here $x>113/20$, and
\begin{align*}
 &\frac{791}{1000}x^4\le M
       \le\frac45x^4+\frac{63}{25}x^2+2,\\
 &\frac{11}{20}x^3-\frac{63}{50}x^2
       \le r\le\frac{11}{20}x^3-\frac54x^2+2
       <\frac{11}{20}x^3,\\
 &\beta\le\frac{633}{1000}x^2-\frac12,\qquad
 Q_{\rm f}(q,k_0)\ge\frac6{25}q^2.
\end{align*}
These follow as in the proof of
Lemma~\ref{prop:analytic-four-fifths-v30}; for the lower
bound on $M$, use $M\ge(y^2-3)/2$ and $x\ge113/20$.
In particular all hypotheses of Lemma~\ref{lem:anchored-Phi-v30}
hold. The same bound on $N_s(q)$ gives
$N_s(q)/Q_{\rm f}(q,k_0)\le40x^{10}$.
Moreover,
\[
 r\log\frac{c}{q-a+1}
 \le-r\log2+r\frac{M+a+1}{q-a+1}
 \le-r\log2+\frac9{20}x.
\]
After substituting the stated bounds on $r,M,a$, the last inequality
reduces to positivity of
\[
 \frac{x^7}{100}+\frac{279}{400}x^6-\frac{1397}{2000}x^5
 +\frac{521}{400}x^4-\frac{22}{5}x^3
 +\frac{123}{50}x^2-\frac9{10}x-12.
\]
For $x\ge5$, group the terms of degrees $6,5$, those of degrees
$4,3$, and those of degrees $2,1,0$; each group is positive.
The logarithmic estimate used in the preceding lemma therefore
gives
\[
 \log\frac{\Psi(q;s,r,M)}{Q_{\rm f}(q,k_0)}\le-F(x),
\]
where
\[
 F(x)=\log2\left(\frac{11}{20}x^3-\frac{1893}{1000}x^2+\frac12\right)
 -\frac9{20}x-1+\frac12\log\frac{20}{11}
 -\frac{23}{2}\log x-\log40.
\]
Its derivative is positive for $x\ge5$, and
\[
 F(113/20)>\frac{18057169}{80000000}>0.
\]
For this last estimate use $\log2>69/100$,
$\log(20/11)>29/50$, $\log(113/20)<433/250$, and
$\log40<37/10$. The first two follow from the positive logarithm
series; the exponential series through degrees $8$ and $9$,
respectively, prove the last two.
Thus Lemma~\ref{lem:anchored-Phi-v30} proves the canonical
classification from $k_0$ to $\lfloor(q-1)/2\rfloor-1$.

It remains to join the geometric range. Even fields have
$K_2(q)\ge\sqrt q$. Odd prime fields are already covered by
Lemma~\ref{lem:prime}. For odd squares, \eqref{eq:geometric-K}
gives $K(q)\ge2\sqrt q/3$. For an odd nonsquare
$q=p^{2m+1}$ put $n=p^m$. If $p=3$, then
$K(q)\ge n=\sqrt q/\sqrt3$. If $p\ge5$, the second term in
\eqref{eq:geometric-K} is at least
\[
 \lfloor\sqrt{q+n+2}\rfloor-n
 \ge\left\lfloor(1-1/\sqrt5)\sqrt q\right\rfloor
 \ge\left\lfloor\frac{11}{20}\sqrt q\right\rfloor.
\]
In every case geometry covers the dimensions below $k_0$.
Lemma~\ref{lem:high} covers the higher dimensions, including its
even-field exception, and Lemma~\ref{thm:dimension-two} supplies
the even-field case $k=2$.
\end{proof}

Only the interval $19683<q<32768$ remains. Its possible prime-power
orders allow the four-fifths estimate to join the geometric bounds.
We now assemble these ranges.

\begin{proof}
The root bound proves that canonical words are deep, and the additional
even-field family is supplied by Lemma~\ref{lem:high}.
Part~(iii) is Lemma~\ref{thm:dimension-two}.
Lemma~\ref{prop:uniform-32768-v30} proves parts~(i)--(ii) when
$q\ge32768$, so suppose $19683<q<32768$.
Odd prime fields are covered by Lemma~\ref{lem:prime}.
There is no even prime power strictly between $19683$ and $32768$.

Suppose first that $q=p^{2m}$ is an odd square. If $m=1$, then
$p>140$, and the geometric endpoint $K(q)=p-1$ is at least
$\lceil4p/5\rceil$. If $m=2$, the only possibility is
$p=13$, giving $q=28561$ and
$K(q)=169-13=156\ge\lceil4\sqrt q/5\rceil=136$.
There is no case $m=3$, since $5^6<19683<32768<7^6$;
there is no case $m=4$, since $3^8<19683<32768<5^8$;
and $m\ge5$ is impossible because $3^{10}>32768$.

Now suppose $q=p^e$ is an odd nonsquare nonprime. If $e=3$,
then $p\in\{29,31\}$, and the Voloch term in
\eqref{eq:geometric-K} gives the respective endpoints
\[
 \left\lfloor\frac{4p^2-29p+31}{16}\right\rfloor=159, 186.
\]
These exceed $\lceil4\sqrt{29^3}/5\rceil=125$ and
$\lceil4\sqrt{31^3}/5\rceil=139$, respectively.
The exponent $e=5$ is impossible since
$7^5<19683<32768<11^5$; $e=7$ is impossible since
$3^7<19683<32768<5^7$; $e=9$ is impossible since
$3^9=19683<32768<5^9$; and $e\ge11$ is impossible since
$3^{11}>32768$.

Thus in every remaining case the geometric range reaches
$\lceil4\sqrt q/5\rceil$. Lemma~\ref{prop:analytic-four-fifths-v30}
covers the dimensions from $\lceil4\sqrt q/5\rceil+1$ through
$\lfloor(q-1)/2\rfloor-1$, and Lemma~\ref{lem:high} covers the
upper range, including its prescribed even-characteristic exception.
This proves the theorem.
\end{proof}

\appendix
\section{Proofs of the corollaries}\label{app:corollaries}
We derive the four consequences stated in the introduction from
Theorem~\ref{thm:main}. We begin with NRC rigidity, then obtain
completeness and projective covering radii, and finally turn to
the two polynomial formulations.

\begin{proof}[Proof of Corollary~\ref{cor:nrc-rigidity}]
Put $m=k$ and $d=q-m$. Choose coordinates on $\mathcal N$ so that
its missing point is $(0:\cdots:0:1)$, and put
$H=(a^j)_{0\le j<m,\,a\in\F_q}$. The field power sums show that
$H$ is a parity-check matrix of $\RS(\F_q,d)$. It has rank $m$,
so every possible extra column is $H\mathbf u^{\mathsf T}$ for
some word $\mathbf u$. By the syndrome correspondence,
$[H\mid H\mathbf u^{\mathsf T}]$ is MDS exactly when $\mathbf u$
is deep. Apply Theorem~\ref{thm:main} with code dimension $d$.
For a canonical interpolation polynomial
$u(X)=\lambda X^d+f(X)$, $\deg f<d$, the syndrome is
$(0,\ldots,0,-\lambda)^{\mathsf T}$, so its projective point is
the missing NRC point. The nonexceptional hypotheses exclude
$d=2$ and $d=q-3$ when $q$ is even.
For $m=3$ in even characteristic, $d=q-3$, and the additional
family \eqref{eq:nucleus-family} has syndrome
$(0,-\lambda,0)^{\mathsf T}$. This is precisely the nucleus
$(0:1:0)$.
\end{proof}

We next use rigidity to obtain completeness, treating the boundary
dimensions directly through the projective code.

\begin{proof}[Proof of Corollary~\ref{cor:prs-radius}]
An NRC in a dimension covered by the first assertion of
Corollary~\ref{cor:nrc-rigidity} is complete: an extra point
would, after deleting one NRC point, have to be that deleted
point. The case $k=2$ is also immediate, since the NRC is all
of $\PG(1,q)$. The dual of $\operatorname{PRS}(\F_q,k)$ is
$\operatorname{PRS}(\F_q,q+1-k)$. D\"ur's criterion therefore
gives the claimed radius from completeness, leaving $k=2$ for
odd $q$ and $k=3$ for even $q$, besides the two even exceptions.

For these remaining nonexceptional cases put $r=2$ when $q$ is
odd and $r=3$ when $q$ is even, and consider a word $(u,z)$ of
length $q+1$. If $u$ is not deep for $\RS(\F_q,r)$, its distance
from the projective code is at most $q-r$. Otherwise
Theorem~\ref{thm:main} gives $u=\lambda X^r+f$ with $\deg f<r$.
Choose an $r$-set $B\subseteq\F_q$ with
\[
 \sum_{b\in B}b=\frac{z-[X^{r-1}]f}{\lambda}.
\]
Such a set is obtained by translating any $r$-set, since the
characteristic does not divide $r$. Then
$g=u-\lambda\prod_{b\in B}(X-b)$ has degree less than $r$,
has $[X^{r-1}]g=z$, and agrees with $u$ on $B$. Thus $(u,z)$
has distance at most $q-r$ from the projective code. Applying
the root bound to a word with finite part $X^r$ gives the reverse
inequality. Hence equality holds and, again by D\"ur's
criterion, the remaining completeness statements. For even $q$,
adjoining the nucleus to a conic gives an MDS extension in
dimension $3$; the dual construction gives one in dimension
$q-1$. These are exactly the exceptional cases, as in
\cite[Proposition~4(2) and equation~(10)]{Kaipa}.
\end{proof}

The functional criterion now translates deep-hole exclusion into
the existence of a solution with distinct coordinates.

\begin{proof}[Proof of Corollary~\ref{cor:symmetric-solutions}]
Put $d=q-k-1$ and consider $C=\RS(\F_q,d)$. Then
$2\le d\le q-3$. In even characteristic, the exclusions
$k\notin\{2,q-3\}$ are exactly $d\notin\{q-3,2\}$.
Define
\[
 u(X)=-\sum_{j=0}^k\delta_jX^{q-1-j}.
\]
Its degree is $q-1-\ell>d$, so Theorem~\ref{thm:main}, applied
with code dimension $d$, shows that $u$ is not deep for $C$.
The functional in \eqref{eq:functional}, now with
$t=q-d-1=k$, satisfies $\mathcal L(X^j)=\delta_j$ for
$0\le j\le k$. By Lemma~\ref{lem:mds}, some $k$-set
$A\subseteq\F_q$ satisfies $\mathcal L(p_A)=0$. Expanding
\[
 p_A(X)=\sum_{r=0}^k(-1)^r e_r(A)X^{k-r}
\]
and using $\delta_j=0$ for $j<\ell$ gives the required equation.
For the last assertion, take $\delta_{k-r}=1$ and all other
coefficients zero, for each $1\le r\le k$.
\end{proof}

The final consequence follows by expressing agreement on a set
of evaluation points in terms of polynomial division.

\begin{proof}[Proof of Corollary~\ref{cor:small-remainder}]
The remainder has degree less than $k$ exactly when a polynomial
of degree less than $k$ agrees with $u$ at all $k+1$ points of
$B$. Such a set exists exactly when $u$ is not a deep hole.
Since $\deg u>k$, Theorem~\ref{thm:main} leaves precisely the
two displayed exceptions.
\end{proof}

\end{document}